\documentclass[
a4paper,
]{article}

\usepackage
[
]
{geometry}

\usepackage{lineno}

\usepackage{setspace}
\usepackage{Settings}

\newcommand{\FBV}{
\bigvee
}

\newcommand{\FBW}{
\bigwedge
}

\newcommand{\FBC}{
\bigcup
}

\title{Completeness of coalition logics with seriality, independence of agents, or determinism}
\author{
Yinfeng Li${}^{1,2}$ and Fengkui Ju${}^{3,4}$\footnote{Corresponding author} \vspace{5pt} \\
{\small {$^1$IRIT-CNRS, University of Toulouse, France}} \\
{\small {$^2$\href{mailto:yinfeng.li@irit.fr}{yinfeng.li@irit.fr}}} \vspace{2.5pt} \\
{\small {$^3$School of Philosophy, Beijing Normal University, China}} \\
{\small {$^4$\href{mailto:fengkui.ju@bnu.edu.cn}{fengkui.ju@bnu.edu.cn}}}
}
\date{
}

\begin{document}

\maketitle

\setlength{\parskip}{0.5em}


\begin{abstract}
\noindent Coalition Logic is a central logic in logical research on strategic reasoning.
In a recent paper, Li and Ju argued that generally, models of Coalition Logic, concurrent game models, have three too strong assumptions: seriality, independence of agents, and determinism. They presented a Minimal Coalition Logic based on general concurrent game models, which do not have the three assumptions.
However, when constructing coalition logics about strategic reasoning in special kinds of situations, we may want to keep some of the assumptions. Thus, studying coalition logics with some of these assumptions makes good sense.
In this paper, we show the completeness of these coalition logics in a uniform way.

\medskip

\noindent \textbf{Keywords:} coalition logics, general concurrent game models, seriality, independence of agents, determinism, completeness

\end{abstract}

\section{Introduction}
\label{sec:Introduction}

\subsection{Coalition Logic and Minimal Coalition Logic}

Coalition Logic $\FCL$ (\cite{pauly_logic_2001,pauly_modal_2002}) is a central logic in logical research on strategic reasoning. Many important logics about strategic reasoning are extensions of $\FCL$: Alternating-time Temporal Logic $\FATL$ (\cite{alur_alternating-time_2002}) is a temporal extension of $\FCL$; Strategy Logic $\FSL$ \cite{mogavero_reasoning_2014} is an extension of $\FCL$ whose language has quantifiers for and names of strategies; and so on. We refer to \cite{benthem_models_2015}, \cite{agotnes_knowledge_2015}, and \cite{sep-logic-power-games} for overviews of the area.

The language of $\FCL$ is a modal language with the featured operator $\Fclo{\FAA} \phi$, indicating \emph{some available joint action of the coalition $\FAA$ ensures $\phi$}.
Its models are \emph{concurrent game models}. Roughly, in a concurrent game model:
there are some states;
there are some agents, who can form coalitions;
at every state, every coalition has some available joint actions;
every joint action of a coalition has some possible outcome states.
The formula $\Fclo{\FAA} \phi$ is true at a state in a concurrent game model if $\FAA$ has an available joint action such that $\phi$ is true at every possible outcome state of the action.

Recently, Li and Ju \cite{li2024minimal} argued that, generally, concurrent game models have three too strong assumptions.
The first one is \emph{seriality}: coalitions always have available joint actions.
The second one is \emph{the independence of agents}: the merge of two available joint actions of two disjoint coalitions is always an available joint action of the union of the two coalitions.
The third one is \emph{determinism}: joint actions of the grand coalition always have a unique outcome.

Here, we briefly mention some arguments from \cite{li2024minimal}.
Games often terminate when reaching specific states. In terminating states, players do not have available actions intuitively. Thus, seriality might not hold.
There are situations where whether a coalition can perform an action is conditional on other agents' actions at the same time. The independence of agents fails in these situations. Here is an example. There are two agents in a room, $a$ and $b$, and only one chair. Agent $a$ can sit, and agent $b$ can sit, but they cannot sit at the same time. 
In many situations, some joint actions of all behaving agents have more than one outcome state. The following example, from \cite{sergot_examples_2014}, can illustrate this.
A vase stands on a table. There is an agent $a$ who can lift or lower the end of the table. If the table tilts, the vase might fall, and if it falls, it might break. Therefore, determinism might not hold.

Based on \emph{general concurrent game models}, which do not have the three assumptions, Li and Ju \cite{li2024minimal} presented a Minimal Coalition Logic $\FMCL$.

\subsection{Logics between Coalition Logic and Minimal Coalition Logic}

Although the three assumptions of concurrent game models are generally too strong, we might want to keep some of them for different reasons when constructing strategic logics.

First, consider seriality. Dropping seriality often makes things complicated. Take Alternating-time Temporal Logic $\FATL$ as an example, where \emph{paths}, that are infinite sequences of states, play an important role. If we drop seriality, we have to consider finite paths and infinite paths, which is somewhat complicated. Scholars often want to keep seriality by assuming that agents can always perform the special action \emph{doing nothing}, which seems not harmful.

Second, consider the independence of agents.
In many situations, agents' actions do not interfere with each other, and the independence of agents holds. For example, in the game of rock-paper-scissors, every player can perform three actions: showing rock, paper, and scissors, which do not interfere with each other.
As mentioned by Sergot \cite{sergot_examples_2014}, the special actions of \emph{attempting to do something} do not interfere with each other.
Some literature, such as \cite{harrenstein_boolean_2001} and \cite{hoek_logic_2005}, deals with agency by \emph{propositional control}: agents are able to control truth values of atomic propositions, and every atomic proposition is controlled by at most one agent. In these settings, the independence of agents holds.

Third, consider determinism.
Typical games are designed such that no matter how players behave, there is always a deterministic outcome. Again, the game of rock-paper-scissors is an example.
In addition, when assuming that \emph{nature} is a player, determinism holds, at least on the macro level.

According to which of the three properties we want to keep, there are eight coalition logics in total. It makes good sense to study the other six coalition logics.

\subsection{Our work}

The completeness $\FCL$ is proved by Pauly \cite{pauly_modal_2002}, and the completeness of $\FMCL$ is shown by Li and Ju \cite{li2024minimal}. The completeness of the other six logics has not been shown yet. In this work, we show the completeness of the eight logics in a uniform way by the \emph{reduction method}.

Generally speaking, to show the completeness of a logic, the reduction method goes through the following four steps.
First, we transform a formula to a conjunction of some \emph{standard disjunctions}. 
Second, we show a \emph{downward validity lemma}, which reduces the validity of a standard disjunction to the validity of some formulas with lower modal depth.
Third, we show an \emph{upward derivability lemma}, which reduces the derivability of a standard disjunction to the derivability of some formulas with lower modal depth.
Fourth, we show the completeness by induction.

Among the four steps, the first and fourth steps are straightforward, the third step is relatively easy, and the second step, showing the downward validity lemma, is crucial. The key part of the second step is to obtain a pointed model that works.

In showing the completeness of the eight coalition logics, we show the downward validity lemmas in a uniform way, where we use a special kind of abstract structure, called blueprints, to get a pointed model that works.

Pauly \cite{pauly_modal_2002} showed the completeness of $\FCL$ by the following approach: first, he showed a \emph{representation theorem} and transformed the semantics for $\FCL$ based on concurrent game models to some kind of \emph{neighborhood semantics}; second, he showed the completeness of $\FCL$ with respect to the neighborhood semantics by the \emph{Henkin method}\footnote{
Goranko, Jamroga, and Turrini \cite{goranko_strategic_2013} pointed out that the representation theorem given by Pauly \cite{pauly_modal_2002} is incorrect, and consequently, the neighborhood semantics given by Pauly \cite{pauly_modal_2002} is incorrect. They offered the correct representation theorem and neighborhood semantics. However, they showed that the neighborhood semantics given by Pauly \cite{pauly_modal_2002} determines the same set of valid formulas as the correct neighborhood semantics.
}.
Our approach differs from this approach.

Goranko and van Drimmelen \cite{goranko_complete_2006} showed the completeness of $\FATL$. As mentioned above, $\FATL$ is a temporal extension of $\FCL$. Their proof contains a proof for the completeness of $\FCL$, which, roughly speaking, is based on (finite) maximal consistent sets. Our proof is different from this proof.

Li and Ju \cite{li2024minimal} proved the completeness of $\FMCL$ by the reduction method, too. The way the downward validity lemmas for the eight coalition logics are shown in this paper is different from how the downward validity lemma for $\FMCL$ is shown in \cite{li2024minimal}. To be more specific, in the latter, we directly constructed a working pointed model (without using blueprints), which heavily depends on the fact that for Minimal Coalition Logic, we do not have to consider seriality, independence of agents, or determinism at all.

The rest of the paper is organized as follows.
In Section \ref{section:Eight coalition logics determined by eight classes of general concurrent game models}, we present the eight coalition logics, including their models, language, semantics, and axiomatic systems.
In Section \ref{section:Completeness of the eight coalition logics}, we show the completeness of the eight logics.
We point out some further work in Section \ref{section:Further work}.

\section{A family of eight coalition logics}
\label{section:Eight coalition logics determined by eight classes of general concurrent game models}

\subsection{General concurrent game models and three properties}

Let $\FAG$ be a finite set of agents. Subsets of $\FAG$ are called \Fdefs{coalitions}. $\FAG$ is called the \Fdefs{grand coalition}.
Let $\FAC$ be a nonempty set of actions. For every coalition $\FCC$, a function $\ja{\FCC}:\FCC \to \FAC $ is called a \Fdefs{joint action} of $\FCC$. The joint actions of $\FAG$ are called \Fdefs{action profiles}.
For every $\FCC \subseteq \FAG$, we define $\FJA_\FCC = \{\sigma_\FCC \mid \sigma_\FCC: \FCC \rightarrow \FAC\}$, which is the set of joint actions of $\FCC$.
We define $\FJA = \bigcup \{\FJA_\FCC \mid \FCC \subseteq \FAG\}$.

\begin{definition}[General concurrent game frames]
\label{definition:??}

A \Fdefs{general concurrent game frame} is a tuple $\F = (\FST, \FAC, \Fout_\FAG)$, where:
\begin{itemize}
\item $\FST$ is a nonempty set of states;
\item $\FAC$ is a nonempty set of actions;
\item $\Fout_\FAG: \FST \times \FJA_{\FAG} \rightarrow \mathcal{P}(\FST)$ is an \Fdefs{outcome function} for $\FAG$.
\end{itemize}

\end{definition}

General concurrent game frames will be simply called \Fdefs{frames} in the sequel.

\begin{definition}[Outcome functions and availability functions]
\label{definition:??}

Let $\F = (\FST, \FAC, \Fout_\FAG)$ be a frame.

For every $\FCC \subseteq \FAG$, define the \Fdefs{outcome function} $\Fout_\FCC$ for $\FCC$ as follows: for all $s\in \FST$ and $\ja{\FCC}\in \FJA_{\FCC}$,
\[
\Fout_\FCC (s,\ja{\FCC}) = \bigcup \{ \Fout_\FAG (s,\ja{\FAG}) \mid \ja{\FAG}\in \FJA_{\FAG} \text{ and } \ja{\FCC}\subseteq \ja{\FAG}\}.
\]

For every $\FCC \subseteq \FAG$, define the \Fdefs{availability function} $\Faja_\FCC$ for $\FCC$ as follows: for all $s \in \FST$ and $\FCC \subseteq \FAG$:
\[\Faja_\FCC (s)= \{\ja{\FCC}\in \FJA_{\FCC} \mid \Fout_\FCC (s,\ja{\FCC}) \neq \emptyset\}.
\]

\end{definition}

The elements of $\aja{s,\FCC}$ are called \Fdefs{available joint actions} of $\FCC$ at $s$.

\begin{definition}[Three properties of frames]
\label{definition:??}

Let $\F = (\FST, \FAC, \Fout_\FAG)$ be a frame. We say:
\begin{itemize}

\item

$\F $ is \Fdefs{serial} if
$\Faja_\FCC (s)\neq \emptyset $ for all $s \in \FST$ and $\FCC \subseteq \FAG$;

\item 

$\F$ is \Fdefs{independent} if
for all $s\in \FST$, $\FCC, \FDD \subseteq \FAG$, $\ja{\FCC} \in \Faja_\FCC (s)$, and $\ja{\FDD} \in \Faja_\FDD (s)$, if $\FCC\cap \FDD = \emptyset $, then $\ja{\FCC} \cup \ja{\FDD} \in \Faja_{\FCC \cup \FDD} (s)$;

\item 

$\F$ is \Fdefs{deterministic} if
for all $s\in \FST$ and $\ja{\FAG}\in \Faja_\FAG (s)$, ${\Fout_\FAG (s, \ja{\FAG})}$ is a singleton.

\end{itemize}

\end{definition}

\newcommand{\FES}{\mathrm{ES}}

We let the three properties correspond to the symbols $\FSX$, $\FIX$, and $\FDX$, respectively.
We let the eight combinations of the three properties correspond to the eight strings $\epsilon$, $\FSX$, $\FIX$, $\FDX$, $\FSX\FIX$, $\FSX\FDX$, $\FIX\FDX$, and $\FSX\FIX\FDX$, respectively.
We use $\FES$ to indicate the set of the eight strings.

For any $\FXX \in \FES$, we say a frame $\F$ is an \Fdefs{$\FXX$-frame} if $\F$ has the properties corresponding to $\FXX$.

Fix a countable set of atomic propositions $\AP$.

\begin{definition}[General concurrent game models]
\label{definition:General concurrent game models}

A \Fdefs{general concurrent game model} is a tuple $\MM = (\FST, \FAC, \Fout_\FAG, \Flab)$, where:
\begin{itemize}

\item

$(\FST, \FAC, \Fout_\FAG)$ is a frame;

\item

$\Flab: \FST \rightarrow \mathcal{P}(\FAP)$ is a \Fdefs{labeling function}.

\end{itemize}

\end{definition}

General concurrent game models will simply be called \Fdefs{models} in the sequel. For any model $\MM$ and state $s$ of $\MM$, $(\MM,s)$ is called a \Fdefs{pointed model}.

For any $\FXX \in \FES$, a model $\MM = (\FST, \FAC, \Fout_\FAG, \Flab)$ is called an \Fdefs{$\FXX$-model} if $(\FST, \FAC, \Fout_\FAG)$ is an $\FXX$-frame.
Note that $\epsilon$-models are models of $\FMCL$ and $\FSX\FIX\FDX$-models are models of $\FCL$.

\subsection{Language}

\begin{definition}[The language $\Phi$]
\label{definition:The language Phi CL}

The language $\Phi$ is defined as follows, where $p$ ranges over $\FAP$ and $\FCC\subseteq \FAG$:
\[
\phi ::=\top \mid p \mid \neg \phi \mid (\phi \wedge \phi) \mid \Fclo{\FCC} \phi
\]

\end{definition}

The formula $\Fclo{\FCC} \phi$ indicates that \emph{some available joint action of $\FCC$ ensures $\phi$}.
Define the propositional connectives $\bot, \lor, \rightarrow$ and $\leftrightarrow$ as usual.
Define the dual $\Fclod{\FCC} \phi$ of $\Fclo{\FCC} \phi$ as $\neg \Fclo{\FCC} \neg \phi$, indicating \emph{every available joint action of $\FAA$ enables $\phi$}.

\subsection{Semantics}

\begin{definition}[Semantics of $\Phi$]
\label{definition:Semantics of Phi CL}

Let $\MM = (\FST, \FAC, \Fout_\FAG, \Flab)$ be a model.

\begin{tabular}{lll}
$\MM, s \Vdash \top$ & & \parbox[t]{28em}{always} \\
$\MM, s \Vdash p$ & $\Leftrightarrow$ & \parbox[t]{28em}{$p \in \Flab (s)$} \\
$\MM, s \Vdash \neg \phi$ & $\Leftrightarrow$ & \parbox[t]{28em}{not $\MM, s \Vdash \phi$} \\
$\MM, s \Vdash \phi \land \psi$ & $\Leftrightarrow$ & \parbox[t]{28em}{$\MM, s \Vdash \phi$ and $\MM, s \Vdash \psi$} \\
$\MM, s \Vdash \Fclo{\FCC} \phi$ & $\Leftrightarrow$ & \parbox[t]{28em}{there is $\ja{\FCC} \in \Faja_\FCC (s)$ such that for all $t \in \Fout_\FCC (s, \ja{\FCC})$, $\MM, t \Vdash \phi$}
\end{tabular}

\end{definition}

It can be verified:

\begin{tabular}{lll}
$\MM, s \Vdash \Fclod{\FCC} \phi$ & $\Leftrightarrow$ & \parbox[t]{28em}{for all $\sigma_\FCC \in \Faja_\FCC (s)$, there is $t \in \Fout_\FCC (s, \sigma_\FCC)$ such that $\MM, t \Vdash \phi$}
\end{tabular}

For any $\FXX \in \FES$, set of formulas $\Gamma$, formula $\phi$, we define the following notions in the expected way: $\phi$ is \Fdefs{$\FXX$-valid} ($\models_\FXL \phi$), $\phi$ is \Fdefs{$\FXX$-satisfiable}, and $\phi$ is a \Fdefs{logical consequence} of $\Gamma$ ($\Gamma \models_\FXL \phi)$.

Let $(\MM,s)$ be a pointed model and $\ja{\FCC}$ be a (possibly unavailable) joint action at $s$.
We use $\ja{\FCC } \achieve_{(\MM, s)} \phi$ to denote $\ja{\FCC }$ \emph{ensures} $\phi$ at $(\MM,s)$, that is, $\MM,t \Vdash \phi$ for all $t \in \Fout_\FCC (s,\ja{\FCC })$.
We use $\ja{\FCC } \allow_{(\MM,s)} \phi$ to denote that $\ja{\FCC }$ \emph{enables} $\phi$ at $(\MM,s)$, that is, $\MM,t \Vdash \phi$ for some $t \in \Fout_\FCC (s,\ja{\FCC })$.

\subsection{Axiomatic systems}

We let the following formulas respectively correspond to the symbols $\FSX$, $\FIX$ and $\FDX$:

\medskip

\begin{tabular}{rl}
Seriality ($\mathtt{A}\text{-}\mathtt{Ser}$): & $\Fclo{\FCC} \top$ \vspace{5pt} \\
Independence of agents ($\mathtt{A}\text{-}\mathtt{IA}$): & $(\Fclo{\FCC} \phi \land \Fclo{\FDD} \psi) \rightarrow \Fclo{\FCC \cup \FDD} (\phi \land \psi)$, where $\FCC \cap \FDD = \emptyset$ \vspace{5pt} \\
Determinism ($\mathtt{A}\text{-}\mathtt{Det}$): & $\Fclo{\FCC} (\phi \lor \psi) \rightarrow (\Fclo{\FCC} \phi \lor \Fclo{\FAG} \psi)$
\end{tabular}

\begin{definition}[Axiomatic systems for $\FXL$]
\label{definition:Axiomatic systems for XL}

For all $\FXX$ in $\FES$, the axiomatic system for $\FXL$ consists of the following axioms and inference rules:

\noindent \emph{Axioms}:

\vspace{5pt}

\begin{tabular}{rl}
Tautologies ($\mathtt{A}\text{-}\mathtt{Tau}$): & all propositional tautologies \vspace{5pt} \\
No absurd available actions ($\mathtt{A}\text{-}\mathtt{NAAA}$): & $\neg \Fclo{\FCC} \bot$ \vspace{5pt} \\
Monotonicity of goals ($\mathtt{A}\text{-}\mathtt{MG}$): & $\Fclo{\emptyset} (\phi \rightarrow \psi) \rightarrow (\Fclo{\FCC} \phi \rightarrow \Fclo{\FCC} \psi)$ \vspace{5pt} \\
Monotonicity of coalitions ($\mathtt{A}\text{-}\mathtt{MC}$): & $\Fclo{\FCC} \phi \rightarrow \Fclo{\FDD} \phi$, where $\FCC \subseteq \FDD$ \vspace{5pt} \\
$\mathtt{A}\text{-}\FXX$: & the formulas corresponding to the elements of $\FXX$
\end{tabular}

\vspace{5pt}

\noindent \emph{Inference rules}:

\vspace{5pt}

\begin{tabular}{rl}
Modus ponens ($\mathtt{MP}$): & $\dfrac{\phi, \phi \rightarrow \psi}
{\psi}$ \vspace{5pt} \\
Conditional necessitation ($\mathtt{CN}$): & $\dfrac{\phi}
{\Fclo{\FCC} \psi \rightarrow \Fclo{\emptyset} \phi}$
\end{tabular}

\end{definition}

For any $\FXX$ in $\FES$, we use $\vdash_\FXL \phi$ to indicate $\phi$ is \Fdefs{derivable} in the system for $\FXL$.

The following fact will be used later, whose proof is skipped:

\begin{fact}
\label{fact:derivative formulas and rules}

The following formulas and rule are derivable in $\FMCL$:

\medskip

\begin{tabular}{rl}
Condition for empty action ($\mathtt{A}\text{-}\mathtt{CEA}$): & $\Fclo{\FCC} \phi \rightarrow \Fclo{\emptyset} \top$ \\
Special independence of agents ($\mathtt{A}\text{-}\mathtt{SIA}$): & $(\Fclo{\emptyset} \phi \land \Fclo{\FCC} \psi) \rightarrow \Fclo{\FCC} (\phi \land \psi)$ \\
Monotonicity ($\mathtt{Mon}$): & from $\phi \rightarrow \psi$, we can get $\Fclo{\FCC} \phi \rightarrow \Fclo{\FDD} \psi$, \\
& where $\FCC \subseteq \FDD$. 
\end{tabular}

\end{fact}

\paragraph{Remarks}

The axiomatic system for $\FMCL$ given in Definition \ref{definition:Axiomatic systems for XL} is shown to be sound and complete with respect to the set of $\epsilon$-valid formulas in \cite{li2024minimal}.
By \cite{li2024minimal}, the axiomatic system for $\FCL$ given in Definition \ref{definition:Axiomatic systems for XL} is equivalent to the sound and complete axiomatic system for $\FCL$ given in \cite{pauly_modal_2002}.
 
The formula $\Fclo{\FCC} \top$ defines the class of frames where $\FCC$ at every state has an available joint action. However, the formula $(\Fclo{\FCC} p \land \Fclo{\FDD} q) \rightarrow \Fclo{\FCC \cup \FDD} (p \land q)$, where $\FCC$ and $\FDD$ are disjoint, does not define the class of frames where at every state, the union of any joint actions of $\FCC$ and $\FDD$ is a joint action of $\FCC \cup \FDD$; the formula $\Fclo{\FCC} (p \lor q) \rightarrow (\Fclo{\FCC} p \lor \Fclo{\FAG} q)$ does not define the class of frames where at every state, every available joint action of $\FAG$ has a unique outcome state.

In the literature, the formula $\neg \Fclo{\emptyset} \neg \phi \rightarrow \Fclo{\FAG} \phi$ is commonly treated as an axiom related to determinism. We want to point out that it is not valid in all deterministic frames, although it is valid in all serial and deterministic frames.

\section{Completeness of the eight coalition logics}
\label{section:Completeness of the eight coalition logics}

\subsection{Our approach: the reduction method}

To show the completeness of an axiomatic system, it suffices to show that every valid formula is derivable in it. We will achieve this by induction on the modal depth of formulas.

The basic step of the induction can be directly done by the classical propositional logic. We will go through the inductive step as follows.

First, we transfer every formula to a normal form, which is a conjunction of some \Fdefs{standard disjunctions}, that are disjunctions meeting some conditions.

Second, we show the \Fdefs{downward validity lemma}: \textit{for every standard disjunction $\phi$, if $\phi$ is valid, then the \Fdefs{validity-reduction condition} of $\phi$ is met.} Here, the validity-reduction condition of $\phi$ concerns the validity of some formulas with lower modal depth than $\phi$.

Third, we show the \Fdefs{upward derivability lemma}: \textit{for every standard disjunction $\phi$, if the \Fdefs{derivability-reduction condition} of $\phi$ is met, then $\phi$ is derivable.} Here, the derivability-reduction condition of $\phi$ is the result of replacing \emph{validity} in the validity-reduction condition of $\phi$ by \emph{derivability}.

We will show the downward validity lemma by showing its contrapositive, called the \Fdefs{upward satisfiability lemma}: \textit{for every \Fdefs{standard conjunction} $\phi'$, if the \Fdefs{satisfiability-reduction condition} of $\phi'$ is met, then $\phi'$ is satisfiable.}
Here, a standard conjunction $\phi'$ is equivalent to the negation of a standard disjunction $\phi$, and the satisfiability-reduction condition of $\phi'$ is equivalent to the negation of the validity-reduction condition of $\phi$.

To show the upward satisfiability lemma, the key is to find a pointed model satisfying a standard conjunction $\phi'$, given that its satisfiability-reduction condition is met. To achieve this, we will use an important notion, called \Fdefs{blueprints}.

\paragraph{Remarks}

Suppose we have proved that a sound logic is complete by this approach. It is easy to see that the other direction of the downward validity lemma and the upward derivability lemma also holds:
\begin{itemize}

\item \textit{For every standard disjunction $\phi$, if the validity-reduction condition of $\phi$ is met, then $\phi$ is valid.}

\item \textit{For every standard disjunction $\phi$, if $\phi$ is derivable, then the derivability-reduction condition of $\phi$ is met.}

\end{itemize}

\noindent This is why we use ``reduction'' in the names of the two lemmas.

For the reduction method, the downward validity lemma is crucial and other steps are relatively easy.

\newcommand{\Fmd}{\mathsf{md}}

\subsection{A normal form lemma}

For every $n \in \mathbb{N}$, the set $\{-x \mid x \in \mathbb{N} \text{ and } 1 \leq x \leq n \}$ is called a \Fdefs{negative indice set} and the set $\{x \mid x \in \mathbb{N} \text{ and } 0 \leq x \leq n \}$ is called a \Fdefs{positive indice set}. Note that negative index sets are allowed to be empty, but positive index sets are not.

We call a disjunction of propositional literals an \Fdefs{elementary disjunction}, and a conjunction of propositional literals an \Fdefs{elementary conjunction}.

\begin{definition}[Standard disjunctions and conjunctions]
~

A \Fdefs{standard disjunction} is a formula in the form $\gamma \vee (\FBW_{i\in \FNI}\Fclo{\FAA_i} \phi_{i} \to \FBV_{j\in \FPI}\Fclo{\FBB_j}\phi_{j})$, where $\gamma$ is an elementary disjunction, $\FNI$ is a negative indice set, and $\FPI$ is a positive indice set such that (1) if $\FNI\neq \emptyset$, then $\Fclo{\FAA_{-1}} \phi_{-1} = \Fclo{\emptyset} \top$, and (2) $\Fclo{\FBB_0}\phi_{0}=\Fclo{\FAG}\bot$.

A \Fdefs{standard conjunction} is a formula in the form $\gamma \wedge \FBW_{i\in \FNI}\Fclo{\FAA_i}\phi_{i}\wedge \FBW_{j\in \FPI}\neg \Fclo{\FBB_j}\phi_{j}$, where $\gamma$ is an elementary conjunction, $\FNI$ is a negative indice set, and $\FPI$ is a positive indice set such that (1) if $\FNI\neq \emptyset$, then $\Fclo{\FAA_{-1}} \phi_{-1} = \Fclo{\emptyset} \top$, and (2) $\neg \Fclo{\FBB_0}\phi_{0}= \neg \Fclo{\FAG}\bot$.

\end{definition}

It is easy to see that the negation of a standard disjunction is equivalent to a standard conjunction, and vice versa.

The special structure of standard disjunctions and conjunctions will make the statement of the downward validity lemma and the upward derivability lemma simpler.

\begin{example}[Standard disjunctions and conjunctions]
\label{example:Standard formulas}

Assume $\FAG = \{a,b\}$.
Let $\gamma = \bot$, $\FNI = \{-1, -2, -3\}$ and $\FPI = \{0,1\}$.
Then,
\[\bot \lor \big(
\big( \Fclo{\FAA_{-1}} \phi_{-1} \land \Fclo{\FAA_{-2}} \phi_{-2} \land \Fclo{\FAA_{-3}} \phi_{-3} \big) \rightarrow \big( \Fclo{\FBB_{0}} \psi_0 \lor \Fclo{\FBB_{1}} \psi_1 \big)
\big)\]
is a standard disjunction, and
\[\neg \bot \land 
\Fclo{\FAA_{-1}} \phi_{-1} \land \Fclo{\FAA_{-2}} \phi_{-2} \land \Fclo{\FAA_{-3}} \phi_{-3} \land \neg \Fclo{\FBB_{0}} \psi_0 \land \neg \Fclo{\FBB_{1}} \psi_1
\]
is a standard conjunction, with respect to $\gamma$, $\FNI$ and $\FPI$, where
$\FAA_{-1} = \emptyset$, $\phi_{-1} = \top$, $\FAA_{-2} = \{b\}$, $\phi_{-2} = q$, $\FAA_{-3} = \{a\}$, $\phi_{-3} = p$, $\FBB_{0} = \FAG$, $\psi_0 = \bot$, $\FBB_{1} = \FAG$, and $\psi_1 = p \land q$.

\end{example}

\begin{definition}[Modal depth of formulas]

Recursively define the \Fdefs{modal depth} $\Fmd (\phi)$ of formulas $\phi$ in $\Phi$ as follows:
\begin{itemize}

\item $\Fmd (\top) = (p) = 0$;

\item $\Fmd (\neg \psi) = \Fmd (\psi)$;

\item $\Fmd (\psi \land \chi) = \max (\Fmd (\psi), \Fmd (\chi))$;

\item $\Fmd (\Fclo{\FCC} \psi) = \Fmd (\psi) + 1$.

\end{itemize}

\end{definition}

\begin{lemma}[Normal form]
\label{lemma:normal-form}

Let $\FXX \in \FES$. For every $\phi \in \Phi$ such that $0 < \Fmd(\phi)$, there is $\phi'$ such that (1) $\vdash_\FXL \phi \leftrightarrow \phi'$, (2) $\phi$ and $\phi'$ have the same modal depth, and (3) $\phi'$ is in the form of $\delta_0 \land \dots \land \delta_k$, where every $\delta_i$ is a standard disjunction.

\end{lemma}

\begin{proof}

Let $\phi \in \Phi$ such that $0 < \Fmd(\phi)$.

First, we can transform $\phi$ to $\psi$ in the conjunctive normal form $\psi_0 \land \dots \land \psi_k$ such that $\psi$ has the same modal depth as $\phi$ and $\vdash \phi \leftrightarrow \psi$.

Fix a $\psi_i$. Note $\psi_i$ is in the form $\alpha_0 \lor \dots \lor \alpha_l$, where every $\alpha_j$ is either a propositional literal, or in the form $\Fclo{\FCC} \beta$, or in the form $\neg \Fclo{\FCC} \beta$.

Assume there is no $\alpha_j$ in the form $\Fclo{\FCC} \beta$. Let $\chi_i = \psi_i \lor \bot \lor \neg \Fclo{\FAG} \bot$. Note $\neg \bot$ and $\neg \Fclo{\FAG} \bot$ are axioms of $\FXL$. Then $\vdash \psi_i \leftrightarrow \chi_i$.
Assume there is $\alpha_j$ in the form $\Fclo{\FCC} \beta$. Let $\chi_i = \psi_i \lor \bot \lor \neg \Fclo{\FAG} \bot \lor \Fclo{\emptyset} \top$. Note $\Fclo{\FCC} \beta \rightarrow \Fclo{\emptyset} \top$ is derivable. Then $\vdash \psi_i \leftrightarrow \chi_i$.
It is easy to see that in the two cases, either $\Fmd(\chi_i) = \Fmd(\psi_i)$, or $\Fmd(\chi_i) = 1$.

It is easy to transform $\chi_i$ to a standard disjunction $\delta_i$ such that $\vdash \chi_i \leftrightarrow \delta_i$ and $\Fmd (\delta_i) = \Fmd (\chi_i)$.

Let $\phi' = \delta_0 \land \dots \land \delta_k$. Then $\vdash \phi \leftrightarrow \phi'$. It is easy to see that $\phi$ and $\phi'$ have the same modal depth. 

\end{proof}

\subsection{Generated submodels}

In this subsection, we define generated submodels and state a result about them, which will be used in a result about blueprints in the next subsection.

\begin{definition}[Generated submodels]

Let $\MM = (\FST, \FAC, \Fout_\FAG, \Flab)$ be a model and $s \in \FST$. Define a model $\MM' = (\FST', \FAC, \Fout'_\FAG, \Flab')$ as follows, called the \Fdefs{generated submodel} of $\MM$ from $s$.
\begin{itemize}

\item

$\FST'$ is the smallest subset of $\FST$ meeting the following conditions:
\begin{itemize}

\item $s \in \FST'$;

\item for every $x \in \FST$, if $x \in \Fout_\FAG (s', \sigma_\FAG)$ for some $s' \in \FST'$ and $\sigma_\FAG \in \FJA_\FAG$, then $x \in \FST'$.

\end{itemize}

\item 

$\Fout'_\FAG$ is the restriction of $\Fout_\FAG$ to $\FST' \times \FJA_{\FAG}$;

\item

$\Flab'$ is the restriction of $\Flab$ to $\FST'$.

\end{itemize}

\end{definition}

The following result indicates that \emph{generated submodels preserve truth values}. 

\begin{fact}

Let $\MM = (\FST, \FAC, \Fout_\FAG, \Flab)$ be a model and $s \in \FST$. Let $\MM' = (\FST', \FAC, \Fout'_\FAG, \Flab')$ be the generated submodel of $\MM$ from $s$. Then for every $\phi \in \Phi$, for every $s' \in \FST'$, $\MM, s' \Vdash \phi$ if and only if $\MM', s' \Vdash \phi$.

\end{fact}

This fact is easy to show by induction, and we skip its proof.

\subsection{Blueprints and their realization}

In this subsection, we define an important notion, blueprints, and prove a result about it.
Intuitively, a blueprint can be viewed as a guide for constructing a model. Later in the proof for the downward validity lemma, we will use blueprints to construct the needed models.

\begin{definition}[Blueprints]

A \Fdefs{blueprint} is a tuple $\pg = (\FAC, \FlistAG)$, where:
\begin{itemize}
\item $\FAC$ is a nonempty set of actions;
\item $\FlistAG: \FJA_{\FAG}\to \powerset{\Phi}$ is a \Fdefs{listing funtion} for $\FAG$.
\end{itemize}

\end{definition}

For any $\sigma_\FAG \in \FJA_\FAG$, intuitively, $\FlistAG (\sigma_\FAG)$ specifies some formulas which $\sigma_\FAG$ can \emph{enable}.

We now generalize list functions to all coalitions and define performable joint actions, which have a similar meaning to available joint actions.

\begin{definition}[Listing functions for and performable joint actions of coalitions]

Let $\pg = (\FAC, \FlistAG)$ be a blueprint.

For all $\FCC \subseteq \FAG$, define the \Fdefs{listing function} $\Flist_\FCC$ for $\FCC$ as follows: for all $\sigma_\FCC \in \FJA_\FCC$,
\[
\Flist_\FCC (\sigma_\FCC) = \bigcup \{\FlistAG (\sigma_\FAG) \mid \sigma_\FAG \in \FJA_\FAG \text{ and } \sigma_\FCC \subseteq \sigma_\FAG\}.
\]

For all $\FCC \subseteq \FAG$, define the \Fdefs{set of performable joint actions} $\FPJA_{\FCC}$ of $\FCC$ as follows:
\[
\FPJA_{\FCC} = \{\sigma_\FCC \in \FJA_\FCC \mid \Flist_\FCC (\ja{\FCC}) \neq \emptyset\}.
\]

We use $\FPJA$ to denote $\bigcup \{\FPJA_\FCC \mid \FCC \subseteq \FAG\}$.

\end{definition}

For every $\FXX \in \FES$, a blueprint should meet some conditions to be used to construct an $\FXX$-model, which is what the following definition is about.

\begin{definition}[Regular blueprints]

Let $\FXX \in \FES$ and $\pg = (\FAC, \FlistAG)$ be a blueprint.

We say $\pg$ is \Fdefs{$\FXX$-regular} if the following conditions are met:
\begin{enumerate}[label=(\arabic*),leftmargin=3.33em]

\item 

for all $\FCC \subseteq \FAG$, $\ja{\FCC} \in \FJA_{\FCC}$, and $\chi \in \Flist_\FCC (\ja{\FCC})$, $\chi$ is $\FXX$-satisfiable;

\item 

if $\FSX \in \FXX$, then $\pg$ is \Fdefs{serial}, that is, $\FPJA_{\FCC} \neq \emptyset$ for all $\FCC \subseteq \FAG$;

\item

if $\FIX \in \FXX$, then $\pg$ is \Fdefs{independent}, that is, $\ja{\FCC}\cup\ja{\FDD}\in \FPJA_{\FCC \cup \FDD}$ for all $\FCC, \FDD \subseteq \FAG$ such that $\FCC \cap \FDD = \emptyset$, $\ja{\FCC}\in \FPJA_{\FCC} $, and $\ja{\FDD}\in \FPJA_{\FDD} $;

\item

if $\FDX \in \FXX$, then $\pg$ is \Fdefs{deterministic}, that is, $\FlistAG (\ja{\FAG})$ is a singleton for all $\ja{\FAG} \in \FPJA_{\FAG}$.

\end{enumerate}

\end{definition}

\begin{definition}[Ensuring functions and enabling functions]
Let $\FXX \in \FES$ and $\pg = (\FAC, \FlistAG)$ be an $\FXX$-regular blueprint.

Define the \Fdefs{ensuring function} $\mathcal{A}_\FXX$ as follows: for all $\FCC \subseteq \FAG$ and $\ja{\FCC} \in \FJA_{\FCC}$,
\[
\all{\FXX}{\ja{\FCC}} = \{\psi \in \Phi \mid \text{ for all } \phi \in \Flist_\FCC (\ja{\FCC}), \phi \vDash_\FXL \psi\}.
\]

Define the \Fdefs{enabling function} $\mathcal{E}_\FXX$ as follows: for all $\FCC \subseteq \FAG$ and $\ja{\FCC} \in \FPJA_{\FCC}$,
\[
\ex{\FXX}{\ja{\FCC}} = \{\psi \in \Phi \mid \text{ for some } \phi \in \Flist_\FCC (\ja{\FCC}), \phi \vDash_\FXL \psi\}.
\]

\end{definition}

Intuitively, $\ja{\FCC}$ \emph{ensures} the formulas in $\all{\FXX}{\ja{\FCC}}$, and \emph{enables} the formulas in $\ex{\FXX}{\ja{\FCC}}$.

\begin{definition}[Realization of blueprints]
\label{definition:??}

Let $\FXX \in \FES$, $\pg = (\FAC_0, \FlistAG)$ be an $\FXX$-regular blueprint, and $\gamma$ be an $\FXX$-satisfiable elementary conjunction.

We say a pointed $\FXX$-model $(\MM,s_0)$, where $\MM = (\FST, \FAC, \Fout_\FAG, \Flab)$, \Fdefs{realizes} $\pg$ and $\gamma$ if the following conditions are met:
\begin{enumerate}[label=(\arabic*),leftmargin=3.33em]

\item 

$\FAC_0 \subseteq \FAC$;

\item 

$\Faja_\FCC (s_0) =\FPJA_{\FCC}$ for every $\FCC \subseteq \FAG$;

\item 

$\MM, s_0 \Vdash \gamma$;

\item 

$\ja{\FCC} \achieve_{(\MM,s_0)} \psi$ for every $\FCC \subseteq \FAG$, $\ja{\FCC} \in \FJA_{\FCC}$ and $\psi \in \all{\FXX}{\ja{\FCC}}$;

\item 

$\ja{\FCC} \allow_{(\MM,s_0)} \psi$ for every $\FCC \subseteq \FAG$, $\ja{\FCC} \in \FJA_{\FCC}$ and $\psi\in\ex{\FXX}{\ja{\FCC}}$.

\end{enumerate}

\end{definition}

\begin{theorem}[Realizability of blueprints]
\label{theorem:Realizability of blueprints}
Let $\FXX \in \FES$. For all $\FXX$-regular blueprint $\pg$ and $\FXX$-satisfiable elementary conjunction $\gamma$, there is a pointed $\FXX$-model $(\MM,s_0)$ realizing them.
\end{theorem}

\begin{proof}

Let $\pg = (\FAC_0, \FlistAG)$ be an $\FXX$-regular blueprint, and $\gamma$ be an $\FXX$-satisfiable elementary conjunction.

\paragraph{Construction of a pointed model}
~

Let $\{(\MM^{\lambda}_{\phi}, s^{\lambda}_{\phi}) \mid \lambda \in \FPJA_{\FAG} \text{ and } \phi \in \listing{\lambda} \}$, where $\MM^{\lambda}_{\phi} = (\FST^{\lambda}_{\phi}, \FAC^{\lambda}_{\phi}, {\Fout_\FAG}^{\lambda}_{\phi}, \Flab^{\lambda}_{\phi})$, be a set of pointed $\FXX$-models meeting the following conditions:
\begin{itemize}

\item 

for every $\lambda \in \FPJA_{\FAG}$ and $\phi \in \listing{\lambda}$, $\MM^{\lambda}_{\phi}, s^{\lambda}_{\phi} \Vdash \phi$;

\item 

all $\FST^{\lambda}_{\phi}$ are pairwise disjoint;

\item 

all $\FAC^{\lambda}_{\phi}$ and $\FAC_0$ are pairwise disjoint.

\end{itemize}

\noindent Note $\{(\MM^{\lambda}_{\phi}, s^{\lambda}_{\phi}) \mid \lambda \in \FPJA_{\FAG} \text{ and } \phi \in \listing{\lambda} \}$ is empty if $\FPJA_\FAG$ is empty.

Let $s_0$ be a state not in any of these models.
Define a pointed model $(\MM,s_0)$, where $\MM = (\FST, \FAC, \Fout_\FAG, \Flab)$, as follows:
\begin{itemize}

\item 

$\FST =\{s_0\} \cup \bigcup \{\FST^{\lambda}_{\phi} \mid \lambda \in \FPJA_{\FAG} \text{ and } \phi \in \listing{\lambda}\}$;

\item

$\FAC =\FAC_{0} \cup \bigcup \{\FAC^{\lambda}_{\phi} \mid \lambda \in \FPJA_{\FAG} \text{ and } \phi \in \listing{\lambda}\}$;

\item 

$\Fout_\FAG: \FST \times \FJA_{\FAG} \to \powerset{\FST}$ such that for all $s \in \FST$ and $\ja{\FAG} \in \FJA_{\FAG}$:
\[
{\Fout_\FAG (s, \ja{\FAG})} =
\begin{cases}
{{\Fout_\FAG}^{\lambda}_{\phi} (s, \ja{\FAG}}) & \text{if } s \in {\FST^{\lambda}_{\phi}} \text{ and } \ja{\FAG}\in {\FJA_\FAG}^\lambda_\phi \\
& \text{for some } \lambda \in \FPJA_{\FAG} \text{ and } \phi \in \listing{\lambda} \\
\{s^{\ja{\FAG}}_{\phi} \mid \phi \in \listing{\ja{\FAG}}\} & \text{if } s = s_0 \text{ and } \ja{\FAG} \in \FPJA_{\FAG} \\
\emptyset & \text{otherwise}
\end{cases}
\]
where ${\FJA_\FAG}^\lambda_\phi = \{\delta: \FAG\to \FAC^{\lambda}_{\phi}\}$;

\item 

$\Flab: \FST \to \powerset{\FAP}$ such that for all $s \in \FST$:
\[
{\Flab (s)}=
\begin{cases}
{\Flab^{\lambda}_{\phi} (s)} & \text{if } s \in {\FST^{\lambda}_{\phi}} \text{ for some } \lambda \in \FPJA_{\FAG}\text{ and } \phi \in \listing{\lambda} \\
\{p \mid \text{$p$ is a conjunct of $\gamma$}\} & \text{if } s = s_0 \\
\end{cases}
\]

\end{itemize}

\paragraph{The constructed pointed model works}
~

We claim $(\MM,s_0)$ is a pointed $\FXX$-model realizing $\pg$ and $\gamma$.

It is easy to verify the first two conditions given in the definition of realization of blueprints:
\begin{enumerate}[label=(\arabic*),leftmargin=3.33em]

\item 

$\FAC_0 \subseteq \FAC$;

\item 

$\Faja_\FCC (s_0) =\FPJA_{\FCC}$ for every $\FCC \subseteq \FAG$.

\end{enumerate}

Then, it is easy to check that $(\MM, s_0)$ is a pointed $\FXX$-model.

Let $\lambda \in \FPJA_{\FAG}$ and $\phi \in \listing{\lambda}$. It can be easily shown the \emph{generated sub-model} of $\MM^{\lambda}_{\phi}$ at $s^{\lambda}_{\phi}$  is also the \emph{generated sub-model} of $\MM$ at $s^{\lambda}_{\phi}$. Then for every $\psi\in \Phi$: $\MM^{\lambda}_{\phi}, s^{\lambda}_{\phi} \Vdash \psi$ if and only if $\MM, s^{\lambda}_{\phi} \Vdash \psi$.

We now show that the last three conditions specified in the definition of realization of blueprints hold.
\begin{enumerate}[leftmargin=3.33em]

\item[(3)]

It is easy to see $\MM,s_0 \Vdash \gamma$.

\item[(4)]

Let $\FCC \subseteq \FAG$, $\ja{\FCC}\in \FJA_{\FCC}$ and $\psi \in \all{X}{\ja{\FCC}}$. We want to show $\ja{\FCC} \achieve_{(\MM,s_0)} \psi$.

Then, for all $\phi \in \Flist_\FCC (\ja{\FCC})$, $\phi \vDash_\FXL \psi$.
Then for all $\ja{\FAG} \in \FJA_{\FAG}$ and $ \phi \in \listing{\ja{\FAG}}$, if $\ja{\FCC}\subseteq \ja{\FAG}$, then $\phi \vDash_{\FXL} \psi$.

Let $s \in \Fout_\FCC (s_0, \sigma_\FCC)$. Then $s \in \Fout_\FAG (s_0, \sigma_\FAG)$ for some $\sigma_\FAG \in \FJA_\FAG$ such that $\sigma_\FCC \subseteq \sigma_\FAG$.
Then, $s = s^{\sigma_\FAG}_{\phi}$ for some $\phi \in \FlistAG (\sigma_\FAG)$. Note $\MM^{\sigma_\FAG}_{\phi},s^{\sigma_\FAG}_{\phi} \Vdash \phi$.
Then $\MM, s^{\sigma_\FAG}_{\phi} \Vdash \phi$.
Then $\MM,s^{\sigma_\FAG}_{\phi} \Vdash \psi$, that is, $\MM,s \Vdash \psi$. Therefore, $\ja{\FCC} \achieve_{(\MM,s_0)} \psi$.

\item[(5)]

Let $\FCC \subseteq \FAG$, $\ja{\FCC}\in \FJA_{\FCC}$ and $\psi \in \ex{\FXX}{\ja{\FCC}}$. We want to show $\ja{\FCC} \allow_{(\MM,s_0)} \psi$.

Then, for some $\phi \in \Flist_\FCC (\ja{\FCC})$, $\phi \vDash_\FXL \psi$.
This implies that for some $\ja{\FAG} \in \FJA_{\FAG}$, $\ja{\FCC}\subseteq \ja{\FAG}$ and $\phi \in \listing{\ja{\FAG}}$.
Thus, $\ja{\FAG} \in \FPJA_{\FAG}$. Then the pointed $\FXX$-model $(\MM^{\sigma_\FAG}_{\phi},s^{\sigma_\FAG}_{\phi})$ is defined and $\MM^{\sigma_\FAG}_{\phi},s^{\sigma_\FAG}_{\phi} \Vdash \phi$.
Then $\MM, s^{\sigma_\FAG}_{\phi} \Vdash \phi $. Then, $\MM, s^{\sigma_\FAG}_{\phi} \Vdash \psi$.
Note $s^{\sigma_\FAG}_{\phi} \in \Fout_\FAG (s_0,\sigma_\FAG)$. Then $s^{\sigma_\FAG}_{\phi} \in \Fout_\FCC (s_0,\sigma_\FCC)$.
Therefore, $\ja{\FCC} \allow_{(\MM,s_0)} \psi$.

\end{enumerate}

\end{proof}

\subsection{Downward validity lemma}

The following notion will be used in stating the downward validity lemma, and also the upward derivability lemma in the next subsection.

\begin{definition}[Neat sets of negative indices]
\label{definition:Neatness}

Let $\FXX \in \FES$ and $\gamma \vee (\FBW_{i\in \FNI}\Fclo{\FAA_i}\phi_i \to \FBV_{j\in \FPI}\Fclo{\FBB_j}\psi_j)$ be a standard disjunction.

For every $\FNI' \subseteq \FNI$, we say $\FNI'$ is \Fdefs{$\FXX$-neat} if the following conditions are met:
\begin{enumerate}[label=(\arabic*),leftmargin=3.33em]

\item 

for all $i,i'\in \FNI'$, if $i\neq i'$, then $\FAA_{i}\cap \FAA_{i'}= \emptyset$;

\item 

if $\FSX \notin \FXX$, then $\FNI' \neq \emptyset $;

\item 

if $\FIX \notin \FXX$, then for all $i,i'\in \FNI'$, if $\FAA_{i}\neq \emptyset $ and $\FAA_{i'}\neq \emptyset $, then $i=i'$.

\end{enumerate}

\end{definition}

\begin{example}[Neat sets of negative indices]
\label{example: neatness}

Assume $\FXX = \FDX$ and $\FAG = \{a,b\}$. Consider the standard disjunction
\[
\bot \lor \big(
\big( \Fclo{\FAA_{-1}} \phi_{-1} \land \Fclo{\FAA_{-2}} \phi_{-2} \land \Fclo{\FAA_{-3}} \phi_{-3} \big) \rightarrow \big( \Fclo{\FBB_{0}} \psi_0 \lor \Fclo{\FBB_{1}} \psi_1 \big)
\big)
\]
with respect to $\gamma = \bot$, $\FNI = \{-1, -2, -3\}$ and $\FPI = \{0,1\}$, given in Example \ref{example:Standard formulas}, where $\FAA_{-1} = \emptyset$, $\phi_{-1} = \top$, $\FAA_{-2} = \{b\}$, $\phi_{-2} = q$, $\FAA_{-3} = \{a\}$, $\phi_{-3} = p$, $\FBB_{0} = \FAG$, $\psi_0 = \bot$, $\FBB_{1} = \FAG$, and $\psi_1 = p \land q$.

It can be verified the following subsets of $\FNI$ are $\FDX$-neat: $\FNI_1 = \{-1\}, \FNI_2 = \{-2\}, \FNI_3 = \{-3\}, \FNI_4 = \{-1, -2\}$, and $\FNI_5 = \{-1, -3\}$.

\end{example}

Let $\gamma \vee (\FBW_{i\in \FNI}\Fclo{\FAA_i}\phi_i \to \FBV_{j\in \FPI}\Fclo{\FBB_j}\psi_j)$ be a standard disjunction. Define $\FPI_0 = \{j \in \FPI \mid \FBB_j = \FAG\}$, called the \Fdefs{set of basic positive indices}.

\medskip

\begin{lemma}[Downward validity]
\label{lemma:Downward validity}

Let $\FXX \in \FES$ and $\gamma \lor (\FBW_{i \in \FNI} \Fclo{\FAA_i} \phi_i \to \FBV_{j \in \FPI} \Fclo{\FBB_j} \psi_j)$ be a standard disjunction.
Then, if $\models_{\FXL} \gamma \lor (\FBW_{i \in \FNI} \Fclo{\FAA_i} \phi_i \to \FBV_{j \in \FPI} \Fclo{\FBB_j} \psi_j)$, then the following condition, called the \Fdefs{$\FXX$-validity-reduction condition}, is met:
\begin{enumerate}[label=(\arabic*),leftmargin=3.33em]

\item 

$\models_{\FXL} \gamma$, or 

\item 

the following two conditions hold:
\begin{enumerate}

\item 

if $\mathsf{D} \notin \FXX$, then there is $\FNI' \subseteq \FNI$ and $j \in \FPI$ such that $\FNI'$ is $\FXX$-neat, $\FBC_{i\in \FNI'} \FAA_i\subseteq \FBB_j$, and $\models_\FXL \FBW_{i\in \FNI'}\phi_i \to \psi_j$;

\item 

if $\mathsf{D} \in \FXX$, then there is $\FNI'\subseteq \FNI$ and $j\in \FPI$ such that $\FNI'$ is $\FXX$-neat, $\FBC_{i\in \FNI'} \FAA_i \subseteq \FBB_{j}$, and $\models_\FXL \FBW_{i \in \FNI'} \phi_i \to (\psi_{j} \lor \FBV_{k \in \FPI_0} \psi_k)$.

\end{enumerate}

\end{enumerate}

\end{lemma}

The following example illustrates this lemma.

\begin{example}
\label{example:11}

Assume $\FXX = \FDX$ and $\FAG = \{a,b\}$. Consider the standard disjunction
\[
\bot \lor \big(
\big( \Fclo{\FAA_{-1}} \phi_{-1} \land \Fclo{\FAA_{-2}} \phi_{-2} \land \Fclo{\FAA_{-3}} \phi_{-3} \big) \rightarrow \big( \Fclo{\FBB_{0}} \psi_0 \lor \Fclo{\FBB_{1}} \psi_1 \big)
\big)
\]
with respect to $\gamma = \bot$, $\FNI = \{-1, -2, -3\}$ and $\FPI = \{0,1\}$, given in Example \ref{example:Standard formulas}, where $\FAA_{-1} = \emptyset$, $\phi_{-1} = \top$, $\FAA_{-2} = \{b\}$, $\phi_{-2} = q$, $\FAA_{-3} = \{a\}$, $\phi_{-3} = p$, $\FBB_{0} = \FAG$, $\psi_0 = \bot$, $\FBB_{1} = \FAG$, and $\psi_1 = p \land q$.

As mentioned in Example \ref{example: neatness}, the following subsets of $\FNI$ are $\FDX$-neat: $\FNI_1 = \{-1\}, \FNI_2 = \{-2\}, \FNI_3 = \{-3\}, \FNI_4 = \{-1, -2\}$, and $\FNI_5 = \{-1, -3\}$.
Note $\FPI_0 = \{0,1\}$.

By this lemma, if $\models_\FDL \phi$, then one of the following holds:
\begin{enumerate}[label=(\arabic*),leftmargin=3.33em]

\item 

$\models_\FDL \gamma$, that is, $\models_\FDL \bot$;

\item 

$\FBC_{i \in \FNI_1} \FAA_i \subseteq \FBB_0$ and $\models \FBW_{i \in \FNI_1} \phi_i \rightarrow (\psi_0 \land \FBW_{j \in \FPI_0} \psi_j)$, that is, $\emptyset \subseteq \FAG$ and $\models_\FDL \top \rightarrow (\bot \land \bot \land (p \land q))$;

\item 

$\FBC_{i \in \FNI_2} \FAA_i \subseteq \FBB_0$ and $\models \FBW_{i \in \FNI_2} \phi_i \rightarrow (\psi_0 \land \FBW_{j \in \FPI_0} \psi_j)$, that is, $\{b\} \subseteq \FAG$ and $\models_\FDL q \rightarrow (\bot \land \bot \land (p \land q))$;

\item 

$\FBC_{i \in \FNI_3} \FAA_i \subseteq \FBB_0$ and $\models \FBW_{i \in \FNI_3} \phi_i \rightarrow (\psi_0 \land \FBW_{j \in \FPI_0} \psi_j)$, that is, $\{a\} \subseteq \FAG$ and $\models_\FDL p \rightarrow (\bot \land \bot \land (p \land q))$;

\item 

$\FBC_{i \in \FNI_4} \FAA_i \subseteq \FBB_0$ and $\models \FBW_{i \in \FNI_4} \phi_i \rightarrow (\psi_0 \land \FBW_{j \in \FPI_0} \psi_j)$, that is, $\{b\} \subseteq \FAG$ and $\models_\FDL (\top \land q) \rightarrow (\bot \land \bot \land (p \land q))$;

\item 

$\FBC_{i \in \FNI_5} \FAA_i \subseteq \FBB_0$ and $\models \FBW_{i \in \FNI_5} \phi_i \rightarrow (\psi_0 \land \FBW_{j \in \FPI_0} \psi_j)$, that is, $\{a\} \subseteq \FAG$ and $\models_\FDL (\top \land p) \rightarrow (\bot \land \bot \land (p \land q))$;

\item 

$\FBC_{i \in \FNI_1} \FAA_i \subseteq \FBB_1$ and $\models \FBW_{i \in \FNI_1} \phi_i \rightarrow (\psi_1 \land \FBW_{j \in \FPI_0} \psi_j)$, that is, $\emptyset \subseteq \FAG$ and $\models_\FDL \top \rightarrow ((p \land q) \land \bot \land (p \land q))$;

\item 

$\FBC_{i \in \FNI_2} \FAA_i \subseteq \FBB_1$ and $\models \FBW_{i \in \FNI_2} \phi_i \rightarrow (\psi_1 \land \FBW_{j \in \FPI_0} \psi_j)$, that is, $\{b\} \subseteq \FAG$ and $\models_\FDL q \rightarrow ((p \land q) \land \bot \land (p \land q))$;

\item 

$\FBC_{i \in \FNI_3} \FAA_i \subseteq \FBB_1$ and $\models \FBW_{i \in \FNI_3} \phi_i \rightarrow (\psi_1 \land \FBW_{j \in \FPI_0} \psi_j)$, that is, $\{a\} \subseteq \FAG$ and $\models_\FDL p \rightarrow ((p \land q) \land \bot \land (p \land q))$;

\item 

$\FBC_{i \in \FNI_4} \FAA_i \subseteq \FBB_1$ and $\models \FBW_{i \in \FNI_4} \phi_i \rightarrow (\psi_1 \land \FBW_{j \in \FPI_0} \psi_j)$, that is, $\{b\} \subseteq \FAG$ and $\models_\FDL (\top \land q) \rightarrow ((p \land q) \land \bot \land (p \land q))$;

\item 

$\FBC_{i \in \FNI_5} \FAA_i \subseteq \FBB_1$ and $\models \FBW_{i \in \FNI_5} \phi_i \rightarrow (\psi_1 \land \FBW_{j \in \FPI_0} \psi_j)$, that is, $\{a\} \subseteq \FAG$ and $\models_\FDL (\top \land p) \rightarrow ((p \land q) \land \bot \land (p \land q))$.

\end{enumerate}

\end{example}

We now show the downward validity lemma.

\begin{proof}
~

\paragraph{Assumptions}
~

Assume the $\FXX$-validity-reduction condition is not met. Then, it can be verified that the following condition, called the \Fdefs{$\FXX$-satisfiability-reduction condition}, is met:
\begin{enumerate}[label=(\arabic*),leftmargin=3.33em]

\item

$\neg \gamma$ is $\FXX$-satisfiable, and 

\item 

the following two conditions hold:
\begin{enumerate}

\item

if $\FDX \notin \FXX$, then for all $\FNI'\subseteq \FNI$ and $j \in \FPI$, if $\FNI'$ is $\FXX$-neat and $\FBC_{i\in \FNI'}\FAA_i\subseteq \FBB_j$, then $\FBW_{i\in \FNI'}\phi_i \wedge \neg \psi_j$ is $\FXX$-satisfiable;

\item 

if $\FDX \in \FXX$, then for all $\FNI'\subseteq \FNI$ and $j \in \FPI$, if $\FNI'$ is $\FXX$-neat and $\FBC_{i\in \FNI'}\FAA_i\subseteq \FBB_{j}$, then $\FBW_{i \in \FNI'} \phi_i \wedge \neg \psi_{j} \wedge \FBW_{k \in \FPI_0} \neg \psi_k$ is $\FXX$-satisfiable.

\end{enumerate}

\end{enumerate}

Let $\gamma'$ be an $\FXX$-satisfiable elementary conjunction, equivalent to $\neg \gamma$.
It suffices to show the standard conjunction $\gamma' \wedge\FBW_{i\in \FNI}\Fclo{\FAA_i}\phi_i \wedge \FBW_{j\in \FPI}\neg \Fclo{\FBB_j}\psi_j$ is $\FXX$-satisfiable.

\medskip

In the sequel, first, we define a blueprint, from which we can get a pointed $\FXX$-model; second, we show the pointed $\FXX$-model satisfies $\gamma' \wedge\FBW_{i\in \FNI}\Fclo{\FAA_i}\phi_i \wedge \FBW_{j\in \FPI}\neg \Fclo{\FBB_j}\psi_j$.

\medskip

To define the blueprint, we need two functions: $\Fsupport$ and $\Fimpeach$.

\paragraph{Two functions $\Fsupport$ and $\Fimpeach$}~

Define two functions $\Fsupport: \FJA \to \powerset{\FNI}$ and $\Fimpeach: \FJA \to \FPI$ as follows:
\begin{itemize}

\item 

for all $\FCC \subseteq \FAG$ and $\ja{\FCC}\in \FJA_{\FCC}$:
\[
\support{\ja{\FCC}} = \{i \in \FNI \mid \FAA_i\subseteq \FCC \text{ and } \act{\FCC}{a} = i\text{ for all } a \in \FAA_i\}.
\]

\item

for all $\FCC \subseteq \FAG$ and $\ja{\FCC}\in \FJA_{\FCC}$:
\[
\impeach{\ja{\FCC}} = \Big( \sum \{ \act{\FCC}{a} \in \FPI \mid a \in \FCC \} \Big) \bmod n,
\]


\noindent where $\sum$ is the sum operation and $n$ is the cardinality of $\FPI$.

\end{itemize}

What follows are some facts about the function $\Fsupport$:
\begin{enumerate}[label=(\arabic*),leftmargin=3.33em]

\item
\label{support well}

\textbf{Claim 1:}
\textbf{for all $\FCC \subseteq \FAG$ and $\ja{\FCC}\in \FJA_{\FCC}$, $\FBC_{i\in \support{\ja{\FCC}}}\FAA_i\subseteq \FCC$, and ${\support{\ja{\FCC}}}$ is both $\FSX\FIX$-neat and $\FSX\FIX\FDX$-neat.}

Let $\FCC \subseteq \FAG$ and $\ja{\FCC}\in \FJA_{\FCC}$.
From the definition of $\Fsupport$, we can see $\FBC_{i\in \support{\ja{\FCC}}}\FAA_i\subseteq \FCC$.
Let $i,i'\in \support{\ja{\FCC}}$ such that $i\neq i'$. To show ${\support{\ja{\FCC}}}$ is both $\FSX\FIX$-neat and $\FSX\FIX\FDX$-neat, it suffices to show $\FAA_i\cap\FAA_{i'} = \emptyset$. Assume $\FAA_i \cap \FAA_{i'} \neq \emptyset$.
Let $a\in \FAA_i \cap \FAA_{i'}$. By the definition of $\Fsupport$, $\sigma_\FCC (a) = i$ and $\act{\FCC}{a} = i'$. Then $i = i'$. We have a contradiction.

\item
\label{support mon}

\textbf{Claim 2:}
\textbf{for all $\FCC, \FCC' \subseteq \FAG$, $\ja{\FCC} \in \FJA_{\FCC}$, and $\ja{\FCC'} \in \FJA_\FCC$, if $\FCC \subseteq \FCC'$ and $\ja{\FCC}\subseteq \ja{\FCC'}$, then ${\support{\ja{\FCC}}} \subseteq \support{\ja{\FCC'}}$.}

Let $\FCC, \FCC' \subseteq \FAG$, $\ja{\FCC} \in \FJA_{\FCC}$, and $\ja{\FCC'} \in \FJA_\FCC$ such that $\FCC \subseteq \FCC'$ and $\ja{\FCC}\subseteq \ja{\FCC'}$. Let $i\in \support{\ja{\FCC}}$. Then $\FAA_i \subseteq \FCC $ and $\act{\FCC}{a} = i$ for all $a \in \FAA_i$. Then $\FAA_i \subseteq \FCC' $ and $\act{\FCC'}{a} = i$ for all $a \in \FAA_i$. Then, $i\in \support{\ja{\FCC'}}$.

\item 
\label{claim:support empty}

\textbf{Claim 3:}
\textbf{if $\FNI \neq \emptyset$, then for all $\sigma_\FAG \in \FJA_\FAG$, $\support{\ja{\FAG}} \neq \emptyset$.}

Assume $\FNI \neq \emptyset$. By the definition of standard disjunctions, there is $i \in \FNI$ such that $\FAA_i = \emptyset$. Then $i\in \support{\ja{\FAG}}$.

\end{enumerate}

\paragraph{Definition of a blueprint $\pg$}~

Define a blueprint $\pg = (\FAC_0, \FlistAG)$ as follows:
\begin{itemize}

\item 

$\FAC = \FNI \cup \FPI$;

\item 

if $\FDX \notin \FXX$, then for every $\ja{\FAG} \in \FJA_{\FAG}$:
\[
\listing{\ja{\FAG}} =
\]
\[
\begin{cases}
\Big\{ \FBW_{i\in \support{\ja{\FAG}}} \phi_i \wedge \neg \psi_j \mid j \in \FPI \text{ and } \FBC_{i\in \support{\ja{\FAG}}} \FAA_i \subseteq \FBB_j \Big\} & \parbox[t]{11em}{
\text{if $\support{\ja{\FAG}}$ is $\FXX$-neat}
} \\
\emptyset & \text{otherwise}
\end{cases}
\]

\item 

if $\FDX \in \FXX$, then for every $\ja{\FAG} \in \FJA_{\FAG}$:
\[
\listing{\ja{\FAG}} =
\]
\[
\begin{cases}
\Big\{ \FBW_{i\in \support{\ja{\FAG}}}\phi_i \wedge \neg \psi_{\Fimpeach(\ja{\FAG})} \wedge \FBW_{k \in \FPI_{0}} \neg \psi_k \Big\} & \parbox[t]{20em}{
\text{if $\support{\ja{\FAG}}$ is $\FXX$-neat and}
} \\
& \parbox[t]{20em}{
\text{$\FBC_{i \in \support{\ja{\FAG}}} \FAA_i \subseteq \FBB_{\Fimpeach(\ja{\FAG})}$}
} \vspace{15pt} \\
\Big\{ \FBW_{i\in \support{\ja{\FAG}}}\phi_i \wedge \neg \psi_{0} \wedge \FBW_{k \in \FPI_{0}} \neg \psi_k \Big\} & \parbox[t]{20em}{
\text{if $\support{\ja{\FAG}}$ is $\FXX$-neat and}
} \\
& \parbox[t]{20em}{
\text{$\FBC_{i \in \support{\ja{\FAG}}} \FAA_i \not \subseteq \FBB_{\Fimpeach(\ja{\FAG})}$}
} \vspace{15pt} \\
\emptyset & \text{otherwise}
\end{cases}
\]

\end{itemize}

\paragraph{The bluepint $\pg$ is $\FXX$-regular}~

The following fact can be easily seen:
\begin{enumerate}[label=(\arabic*),leftmargin=3.33em]
\setcounter{enumi}{3}

\item
\label{claim:nonempty neat}

\textbf{Claim 4:}
\textbf{For every $\sigma_\FAG \in \FJA_\FAG$, $\Flist_\FAG (\sigma_\FAG) \neq \emptyset$ if and only if $\Fsupport (\sigma_\FAG)$ is $\FXX$-neat.}

\end{enumerate}

We claim $\pg$ is $\FXX$-regular.
\begin{enumerate}[label=(\arabic*),leftmargin=3.33em]

\item 

\textbf{For all $\ja{\FAG}\in \FJA_{\FAG}$ and $\chi \in \listing{\ja{\FAG}}$, $\chi $ is $\FXX$-satisfiable.}

Let $\ja{\FAG}\in \FJA_{\FAG}$ and $\chi \in \listing{\ja{\FAG}}$.

\begin{itemize}

\item

Assume $\FDX \notin \FXX$. Then there is $j\in \FPI$ such that $\chi =\FBW_{i\in \support{\ja{\FAG}}}\phi_{i}\wedge \neg \psi_j$, where $\support{\ja{\FAG}}$ is $\FXX$-neat and $\FBC_{i\in\support{\ja{\FAG}}}\FAA_{i} \subseteq \FBB_j$.
By the $\FXX$-satisfiability-reduction condition, $\chi$ is $\FXX$-satisfiable.

\item 

Assume $\FDX \in \FXX$. 
Then either (a) $\chi =\FBW_{i\in \support{\ja{\FAG}}} \phi_{i} \wedge \neg \psi_{\Fimpeach(\ja{\FAG})} \wedge \FBW_{k \in \FPI_0} \neg \psi_k$, where $\support{\ja{\FAG}}$ is $\FXX$-neat and $\FBC_{i\in\support{\ja{\FAG}}}\FAA_{i}\subseteq \FBB_{\Fimpeach(\ja{\FAG})}$, or
(b) $\chi =\FBW_{i\in \support{\ja{\FAG}}} \phi_{i} \wedge \neg \psi_{0} \wedge \FBW_{k \in \FPI_0} \neg \psi_k$, where $\support{\ja{\FAG}}$ is $\FXX$-neat and $\FBC_{i\in\support{\ja{\FAG}}}$ $\FAA_{i} \not \subseteq \FBB_{\Fimpeach(\ja{\FAG})}$.
Assume (a). By the $\FXX$-satisfiablity-redution-condition, $\chi$ is $\FXX$-satisfiable.
Assume (b). Note $\FBB_0 = \FAG$. Then $\FBC_{i\in\support{\ja{\FAG}}} \FAA_{i} \subseteq \FBB_0$. By the $\FXX$-satisfiablity-redution-condition, $\chi$ is $\FXX$-satisfiable.

\end{itemize}

\item

\textbf{If $\FSX \in \FXX$, then $\pg$ is serial.} 

Assume $\FSX \in \FXX$. Let $\FCC \subseteq \FAG$. It suffices to show $\FPJA_{\FCC}\neq \emptyset$.
Let $\ja{\FCC}\in \FJA_{\FCC}$ such that $\act{\FCC}{a}=0$ for all $a\in \FCC$. It suffices to show $\Flist_\FCC (\sigma_\FCC) \neq \emptyset$.
Let $\ja{\FAG}\in \FJA_{\FAG}$ such that $\act{\FAG}{a}=0$ for all $a\in \FAG$. Then $\ja{\FCC}\subseteq \ja{\FAG}$. It suffices to show $\listing{\ja{\FAG}} \neq \emptyset$.
It can be checked $\support{\ja{\FAG}}=\{i\in\FNI\mid \FAA_{i}=\emptyset \}$ and $\support{\ja{\FAG}}$ is $\FXX$-neat. By Claim \ref{claim:nonempty neat}, $\listing{\ja{\FAG}} \neq \emptyset$.

\item 

\textbf{If $\FIX \in \FXX$, then $\pg$ is independent.} 

Assume $\FIX \in \FXX$. Let $\FCC, \FDD \subseteq \FAG$ such that $\FCC\cap\FDD=\emptyset$, $\ja{\FCC}\in \FPJA_{\FCC}$, and $\ja{\FDD}\in \FPJA_{\FDD}$. We want to show $\ja{\FCC}\cup \ja{\FDD} \in \FPJA_{\FCC\cup \FDD}$. It suffices to show $\Flist_{\FCC \cup \FDD} (\sigma_\FCC \cup \sigma_\FDD) \neq \emptyset$.
Let $\ja{\FAG} \in \FJA_{\FAG}$ such that $\ja{\FCC} \cup \ja{\FDD} \subseteq \ja{\FAG}$ and $\act{\FAG}{a}=0$ for all $a\in \FAG-(\FCC\cup \FDD)$. It suffices to show $\Flist_\FAG (\sigma_\FAG) \neq \emptyset$.

By Claim \ref{support well}, $\support{\ja{\FAG}}$ is both $\FSX\FIX$-neat and $\FSX\FIX\FDX$-neat.
\begin{itemize}

\item 

Assume $\FSX \in \FXX$. Then, it can be checked $\support{\ja{\FAG}}$ is $\FXX$-neat. By Claim \ref{claim:nonempty neat}, $\listing{\ja{\FAG}} \neq \emptyset$.

\item 

Assume $\FSX \notin \FXX$.

\begin{itemize}

\item 

Assume $\support{\ja{\FAG}}\neq \emptyset$. It can be verified $\support{\ja{\FAG}}$ is $\FXX$-neat. By Claim \ref{claim:nonempty neat}, $\listing{\ja{\FAG}} \neq \emptyset$.

\item 

Assume $\support{\ja{\FAG}}=\emptyset$. We want to show that this case is impossible.
By Claim \ref{claim:support empty}, $\FNI = \emptyset$.
Therefore, $\support{\lambda_\FAG} = \emptyset$ for all $\lambda_\FAG \in \FJA_{\FAG}$.
Note $\ja{\FCC}\in \FPJA_{\FCC}$. Then $\Flist_\FCC (\sigma_\FCC) \neq \emptyset$.
Then there is $\lambda_\FAG \in \FJA_{\FAG}$ such that $\sigma_\FCC \subseteq \lambda_\FAG$ and $\FlistAG (\lambda_\FAG) \neq \emptyset$. By Claim \ref{claim:nonempty neat}, $\support{\lambda_\FAG}$ is $\FXX$-neat. Note $\FSX \notin \FXX$. Then $\support{\lambda_\FAG} \neq \emptyset$. We have a contradiction.

\end{itemize}

\end{itemize}

\item 

\textbf{If $\FDX \in \FXX$, then $\pg$ is deterministic.}

Assume $\FDX \in \FXX$. From the definition of $\FlistAG$, we can see that for every $\sigma_\FAG \in \FPJA_\FAG$, $\Flist_\FAG (\sigma_\FAG)$ is a singleton. Then, $\pg$ is deterministic.

\end{enumerate} 

\paragraph{The standard conjunction $\gamma' \wedge\FBW_{i\in \FNI}\Fclo{\FAA_i}\phi_i \wedge \FBW_{j\in \FPI}\neg \Fclo{\FBB_j}\psi_j$ is $\FXX$-satisfiable}~

By Theorem \ref{theorem:Realizability of blueprints}, there is a pointed $\FXX$-model $(\MM,s)$, where $\MM = (\FST, \FAC, \Fout_\FAG, \Flab)$, realizing $\pg$ and $\gamma'$.
We claim $\MM,s \Vdash \gamma' \land \FBW_{i\in \FNI}\Fclo{\FAA_i} \phi_i\wedge \FBW_{j\in \FPI}\neg \Fclo{\FBB_j} \psi_j$.

We show this in two steps.

\paragraph{Step one}~

Let $i\in \FNI$. We want to show $\MM,s \Vdash \Fclo{\FAA_i} \phi_i$.
Let $\ja{\FAA_{i}} \in \FJA_{\FAA_i}$ such that $\act{\FAA_{i}}{a} = i$ for all $a\in \FAA_{i}$. It suffices to show $\sigma_{\FAA_i} \in \FPJA_{\FAA_i}$ and $\ja{\FAA_{i}}\achieve_{(\MM, s)} \phi_i$.

Let $\ja{\FAG} \in \FJA_\FAG$ such that $\act{\FAG}{a} = i$ for all $a\in \FAG$. Then $\ja{\FAA_{i}}\subseteq \ja{\FAG}$. It can be checked $\support{\ja{\FAG}} = \{i'\in \FNI\mid \FAA_{i'} = \emptyset \text{ or } i' = i\}$ and $\support{\ja{\FAG}}$ is $\FXX$-neat.
By Claim \ref{claim:nonempty neat}, $\listing{\ja{\FAG}} \neq \emptyset$. Then $\Flist_{\FAA_i} (\sigma_{\FAA_i}) \neq \emptyset$. Then $\sigma_{\FAA_i} \in \FPJA_{\FAA_i}$.

We claim $\phi_i\in \all{\FXX}{\ja{\FAA_{i}}}$. Let $\chi \in \Flist_{\FAA_i} (\sigma_{\FAA_i})$. 
Then $\chi \in \FlistAG (\lambda_\FAG)$ for some $\lambda_\FAG \in \FJA_{\FAG}$ such that $\ja{\FAA_{i}} \subseteq \lambda_\FAG$. Note $i\in \support{\ja{\FAA_i}}$. By Claim \ref{support mon}, $i\in \support{\lambda_\FAG}$. From the definition of $\FlistAG$, we can see $\phi_{i}$ is a conjunct of $\chi$. Then $\chi \vDash_\FXL \phi_{i}$. Then $\phi_i\in \all{\FXX}{\ja{\FAA_{i}}}$.

By theorem \ref{theorem:Realizability of blueprints}, $\ja{\FAA_{i}}\achieve_{(\MM, s)} \phi_i$.

\paragraph{Step two}~

Let $j \in \FPI$. We want to show $\MM,s \Vdash \neg \Fclo{\FBB_j} \psi_j$. Let $\ja{\FBB_j} \in \FPJA_{\FBB_j}$. It suffices to show $\ja{\FBB_j}\allow_{(\MM, s)} \neg \psi_j$.

Note $\Flist_{\FBB_j} (\ja{\FBB_j}) \neq \emptyset$.
Then there is $\lambda_\FAG \in \FPJA_{\FAG}$ such that $\ja{\FBB_j} \subseteq \lambda_\FAG$ and $\FlistAG (\lambda_\FAG) \neq \emptyset$. By Claim \ref{claim:nonempty neat}, $\support{\lambda_\FAG}$ is $\FXX$-neat.
By Claim \ref{support mon}, $\support{\ja{\FBB_j}} \subseteq \support{\lambda_\FAG}$. \Fblue{We claim $\support{\ja{\FBB_j}}$ is $\FXX$-neat. It is easy to check that the first and third conditions given in Definition \ref{definition:Neatness}, the definition of neatness, hold for $\support{\ja{\FBB_j}}$.

Assume $\FSX \in \FXX$. Clearly, $\support{\ja{\FBB_j}}$ meets the second condition given in Definition \ref{definition:Neatness}. Then $\support{\ja{\FBB_j}}$ is $\FXX$-neat.

Assume $\FSX \notin \FXX$. Note $\support{\lambda_\FAG}$ is $\FXX$-neat. Then $\support{\lambda_\FAG} \neq \emptyset$. Then $\FNI \neq \emptyset$ and $-1 \in \FNI$. Note $\FAA_{-1} = \emptyset$. It is easy to check $-1 \in \support{\ja{\FBB_j}}$. Then the second condition given in Definition \ref{definition:Neatness} holds for $\support{\ja{\FBB_j}}$. Then $\support{\ja{\FBB_j}}$ is $\FXX$-neat.
}

We consider two possible cases, respectively.

\medskip

\textbf{Assume $\FDX \notin \FXX$.} Let $\ja{\FAG} \in \FJA_\FAG$ such that $\ja{\FBB_j} \subseteq \ja{\FAG}$ and $\act{\FAG}{\overline{b}}=j$ for all $\overline{b}\in \FAG-\FBB_j$.

We claim $\support{\ja{\FAG}}=\support{\ja{\FBB_j}}$. By Claim \ref{support mon}, $\support{\ja{\FBB_j}}\subseteq \support{\ja{\FAG}}$. Let $i \in \support{\ja{\FAG}}$. Then $\act{\FAG}{a}=i$ for all $a\in \FAA_{i}$.
\begin{itemize}

\item 

Assume $\FAA_{i}\subseteq \FBB_j$. Note $\ja{\FBB_j}\subseteq \ja{\FAG}$. Then $\act{\FBB_j}{a}=i$ for all $a\in \FAA_{i}$. Then $i\in \support{\ja{\FBB_j}}$.

\item 

Assume $\FAA_{i}\nsubseteq \FBB_j$. We want to show that this case is impossible. Let $\overline{b}\in \FAA_{i}-\FBB_j$. Then $\overline{b}\in \FAG-\FBB_j$. Then $\act{\FAG}{\overline{b}}=j$.
Note $\overline{b} \in \FAA_i$. Then $\act{\FAG}{\overline{b}}=i$. Then $j = i$. Note $i \in \FNI$ and $j \in \FPI$. Then $\FNI \cap \FPI \neq \emptyset$. We have a contradiction.

\end{itemize}

Then $\support{\ja{\FAG}}=\support{\ja{\FBB_j}}$.

Note $\support{\ja{\FBB_j}}$ is $\FXX$-neat. Then $\support{\ja{\FAG}}$ is $\FXX$-neat.
By Claim \ref{support well}, $\FBC_{i\in \support{\FBB_j}} \FAA_i$ $\subseteq \FBB_j$. Then $\FBC_{i\in \support{\sigma_\FAG}} \FAA_i \subseteq \FBB_j$.
Then $\FBW_{i \in \support{\ja{\FAG}}}\phi_i\wedge \neg \psi_j$ is in $ \listing{\ja{\FAG}}$.
Note $\ja{\FBB_j} \subseteq \ja{\FAG}$. Then $\FBW_{i \in \support{\ja{\FAG}}}\phi_i\wedge \neg \psi_j$ is in $ \Flist_{\FBB_j} (\ja{\FBB_j})$.
Clearly, $\FBW_{i \in \support{\ja{\FAG}}}\phi_i\wedge \neg \psi_j \vDash_\FXL \neg \psi_j$.
Then $\neg \psi_j \in \ex{\FXX}{\ja{\FBB_j}}$.

By Theorem \ref{theorem:Realizability of blueprints}, $\ja{\FBB_j} \allow_{(\MM,s)} \neg \psi_j$.

\medskip

\textbf{Assume $\FDX \in \FXX$.} Consider whether $\FBB_j=\FAG$.
\begin{itemize}

\item

Assume $\FBB_j=\FAG$. Then $j\in\FPI_0$.
Note $\FlistAG (\lambda_\FAG) \neq \emptyset$. Let $\chi \in \FlistAG (\lambda_\FAG)$.
From the definition of $\Flist_\FAG$, we can see $\neg \psi_j$ is a conjunct of $\chi$. Clearly, $\chi \vDash_\FXL \neg \psi_j$.
Note $\sigma_{\FBB_j} \subseteq \lambda_\FAG$. Then $\chi \in \Flist_{\FBB_j} (\sigma_{\FBB_j})$. Then $\neg \psi_j \in \ex{\FXX}{\ja{\FBB_j}}$.

By Theorem \ref{theorem:Realizability of blueprints}, $\ja{\FBB_j} \allow_{(\MM,s)} \neg \psi_j$.

\item

Assume $\FBB_j\neq \FAG$. Then there is $\overline{b}\in \FAG-\FBB_j$. Let $\ja{\FAG} \in \FJA_\FAG$ meeting the following conditions:
\begin{enumerate}[label=(\arabic*)]

\item

$\act{\FAG}{b}=\act{\FBB_j}{b}$ for all $b\in \FBB_j$;

\item

$\act{\FAG}{a}=0$ for all $a \in \FAG-\FBB_j$ such that $a\neq \overline{b}$.

\item 

$\act{\FAG}{\overline{b}}=
\begin{cases}
j-\Fimpeach(\ja{\FBB_j}) & \text{if } j \geq \Fimpeach(\ja{\FBB_j}) \\
j - \Fimpeach(\ja{\FBB_j})+n & \text{otherwise}
\end{cases}$

where $n$ is the cardinality of $\FPI$.

\end{enumerate}

Note $j < n$. It is easy to see $\sigma_\FAG (\overline{b}) \in \FPI$.
Clearly, $\ja{\FBB_j}\subseteq \ja{\FAG}$.

We claim $\Fimpeach(\ja{\FAG}) = j$.

Note $\Fimpeach(\ja{\FAG}) = \sum \{\sigma_\FAG (a) \in \FPI \mid a \in \FAG\} \bmod n = \Big( \sum \{\sigma_\FAG (a) \in \FPI \mid a \in \FBB_j\} + \act{\FAG}{\overline{b}} \Big) \bmod n$.
It suffices to show $\bigg( \Big( \sum \{\sigma_\FAG (a) \in \FPI \mid a \in \FBB_j\} \bmod n \Big) + \Big( \act{\FAG}{\overline{b}} \bmod n \Big) \bigg) \bmod n = j$.

Assume $j \geq \Fimpeach(\ja{\FBB_j})$.
$\Big( \sum \{\sigma_\FAG (a) \in \FPI \mid a \in \FBB_j\} \bmod n \Big) + \Big( (j - \Fimpeach(\ja{\FBB_j}))$ $\bmod n \Big) = \Fimpeach(\ja{\FBB_j}) + (j - \Fimpeach(\ja{\FBB_j})) = j$.
$j \bmod n = j$.
Thus, $\Fimpeach(\ja{\FAG})$ $= j$.

Assume $j < \Fimpeach(\ja{\FBB_j})$.
$\Big( \sum \{\sigma_\FAG (a) \in \FPI \mid a \in \FBB_j\} \bmod n \Big) + \Big( (j - \Fimpeach(\ja{\FBB_j}) + n) \bmod n \Big) = \Fimpeach(\ja{\FBB_j}) + (j - \Fimpeach(\ja{\FBB_j}) + n) = j + n$.
$j + n \bmod n = j$.
Then, $\Fimpeach(\ja{\FAG}) = j$.

We claim $\support{\ja{\FAG}}=\support{\ja{\FBB_j}}$. By Claim \ref{support mon}, $\support{\ja{\FBB_j}}\subseteq \support{\ja{\FAG}}$. Let $i \in \support{\ja{\FAG}}$. Then $\act{\FAG}{a}=i$ for all $a\in \FAA_{i}$.
\begin{itemize}

\item 

Assume $\FAA_{i}\subseteq \FBB_j$. Note $\ja{\FBB_j}\subseteq \ja{\FAG}$. Then $\act{\FBB_j}{a}=i$ for all $a\in \FAA_{i}$. Then $i\in \support{\ja{\FBB_j}}$.

\item 

Assume $\FAA_{i}\nsubseteq \FBB_j$. We want to show this case is impossible. Let $c \in \FAA_{i}-\FBB_j$. Then $\act{\FAG}{c}=i$. 
Note $c \in \FAG-\FBB_j$. Then $\act{\FAG}{c} \in \FPI$, that is, $i \in \FPI$. We have a contradiction.

\end{itemize}

Then $\support{\ja{\FAG}}=\support{\ja{\FBB_j}}$.

Note $\support{\ja{\FBB_j}}$ is $\FXX$-neat. Then $\support{\ja{\FAG}}$ is $\FXX$-neat.
Note $\FBC_{i\in \support{\ja{\FBB_j}}}\FAA_i\subseteq \FBB_j$. Then $\FBC_{i\in \support{\ja{\FAG}}}\FAA_i \subseteq \FBB_j$.
Therefore, $\FBW_{i\in \support{\ja{\FBB_j}}}\phi_i \wedge \neg \psi_j \wedge \FBW_{k \in \FPI_0} \neg \psi_{k} \in \listing{\ja{\FAG}}$.
Note $\sigma_{\FBB_j} \subseteq \sigma_\FAG$. Then $\FBW_{i\in \support{\ja{\FBB_j}}}\phi_i \wedge \neg \psi_j \wedge \FBW_{k \in \FPI_0} \neg \psi_{k} \in \Flist_{\FBB_j} (\ja{\FBB_j})$.
Clearly, $\FBW_{i\in \support{\ja{\FBB_j}}}\phi_i \wedge \neg \psi_j \wedge \FBW_{k \in \FPI_0} \neg \psi_{k} \models_\FXL \neg \psi_j$.
Then $\neg \psi_j\in \ex{X}{\ja{\FBB_j}}$.

By Theorem \ref{theorem:Realizability of blueprints}, $\ja{\FBB_j} \allow_{(\MM,s)} \neg \psi_j$.

\end{itemize}

\end{proof}

\subsection{Upward derivability lemma}

\begin{lemma}[Upward derivability]
\label{lemma:Upward derivability}

Let $\FXX \in \FES$ and $\gamma \lor (\FBW_{i \in \FNI} \Fclo{\FAA_i} \phi_i \to \FBV_{j \in \FPI} \Fclo{\FBB_j} \psi_j)$ be a standard disjunction.

Then, $\vdash_{\FXL} \gamma \lor (\FBW_{i \in \FNI} \Fclo{\FAA_i} \phi_i \to \FBV_{j \in \FPI} \Fclo{\FBB_j} \psi_j)$, if the following condition, called the \Fdefs{$\FXX$-derivability-reduction condition}, is met:
\begin{enumerate}[label=(\arabic*),leftmargin=3.33em]

\item 

$\vdash_{\FXL} \gamma$, or 

\item 

the following two conditions hold:
\begin{enumerate}

\item

if $\mathsf{D} \notin \FXX$, then there is $\FNI' \subseteq \FNI$ and $j \in \FPI$ such that $\FNI'$ is $\FXX$-neat, $\FBC_{i\in \FNI'} \FAA_i\subseteq \FBB_j$, and $\vdash_\FXL \FBW_{i\in \FNI'}\phi_i \to \psi_j$;

\item 

if $\mathsf{D} \in \FXX$, then there is $\FNI'\subseteq \FNI$ and $j\in \FPI$ such that $\FNI'$ is $\FXX$-neat, $\FBC_{i\in \FNI'} \FAA_i \subseteq \FBB_{j}$, and $\vdash_\FXL \FBW_{i \in \FNI'} \phi_i \to (\psi_{j} \lor \FBV_{k \in \FPI_0} \psi_k)$.

\end{enumerate}

\end{enumerate}

\end{lemma}

\begin{proof}~

Assume $\vdash_\FXL \gamma$. It is easy to see $\vdash_\FXL \gamma \lor (\FBW_{i \in \FNI} \Fclo{\FAA_i} \phi_i \to \FBV_{j \in \FPI} \Fclo{\FBB_j} \psi_j)$.

Assume the following two conditions hold:
\begin{enumerate}[label=(\arabic*),leftmargin=3.33em]

\item 

if $\FDX \notin \FXX$, then there is $\FNI'\subseteq \FNI$ and $j\in \FPI$ such that $\FNI'$ is $\FXX$-neat, $\FBC_{i\in \FNI'}\FAA_i\subseteq \FBB_j$, and $\vdash_\FXL \FBW_{i\in \FNI'}\phi_i \to \psi_j$;

\item 

if $\FDX \in \FXX$, then there is $\FNI'\subseteq \FNI$ and $j\in \FPI$ such that $\FNI'$ is $\FXX$-neat, $\FBC_{i\in \FNI'}\FAA_i\subseteq \FBB_{j}$, and $\vdash_\FXL \FBW_{i\in \FNI'}\phi_i \to (\psi_{j}\vee \FBV_{k \in \FPI_0} \psi_k)$.

\end{enumerate}

We respectively consider whether $\FDX \in \FXX$.
\begin{enumerate}[label=(\arabic*),leftmargin=3.33em]

\item 

Assume $\FDX \notin \FXX$. Then, there is $\FNI'\subseteq \FNI$ and $j\in \FPI$ such that $\FNI'$ is $\FXX$-neat, $\FBC_{i\in \FNI'}\FAA_i\subseteq \FBB_j$, and $\vdash_\FXL \FBW_{i\in \FNI'}\phi_i \to \psi_j$. By Rule $\mathtt{Mon}$, $\vdash_\FXL \Fclo{\FBC_{i \in \FNI'} \FAA_i} \FBW_{i \in \FNI'} \phi_i \rightarrow \Fclo{\FBB_j} \psi_{j}$. We respectively consider whether $\FSX$ is in $\FXX$ and whether $\FIX$ is in $\FXX$.

\begin{enumerate}

\item 

Assume $\FSX \notin \FXX$ and $\FIX \notin \FXX$.
Note in this case, $\FNI'$ is $\epsilon$-neat: $\FNI'$ is not empty, and for all $i,i'\in \FNI'$, if $\FAA_{i}\neq \emptyset $ and $\FAA_{i'}\neq \emptyset$, then $i=i'$.
By repeated applications of Axiom $\mathtt{A}\text{-}\mathtt{SIA}$, we can get $\vdash_\FXL \FBW_{i \in \FNI'} \Fclo{\FAA_i} \phi_i \rightarrow \Fclo{\FBC_{i \in \FNI'} \FAA_i} \FBW_{i \in \FNI'} \phi_i$.
Clearly, $\vdash_\FXL \FBW_{i \in \FNI} \Fclo{\FAA_i} \phi_i \rightarrow \FBW_{i \in \FNI'} \Fclo{\FAA_i} \phi_i$, and $\vdash_\FXL \Fclo{\FBB_j} \psi_{j} \rightarrow \FBV_{j \in \FPI} \Fclo{\FBB_j} \psi_j$.
Then $\vdash_\FXL \FBW_{i \in \FNI} \Fclo{\FAA_i} \phi_i \rightarrow \FBV_{j \in \FPI} \Fclo{\FBB_j} \psi_j$.
Then, $\vdash_\FXL \gamma \lor (\FBW_{i \in \FNI} \Fclo{\FAA_i} \phi_i \rightarrow \FBV_{j \in \FPI} \Fclo{\FBB_j} \psi_j)$.

\item

Assume $\FSX \notin \FXX$ and $\FIX \in \FXX$.
Note in this case, Axiom $\mathtt{A}\text{-}\mathtt{IA}$ is an axiom of $\FXL$, and $\FNI'$ is $\FIX$-neat: $\FNI'$ is not empty, and for all $i,i'\in \FNI'$, if $i\neq i'$, then $\FAA_{i}\cap \FAA_{i'}= \emptyset$.
By Axiom $\mathtt{A}\text{-}\mathtt{IA}$, we can get $\vdash_\FXL \FBW_{i \in \FNI'} \Fclo{\FAA_i} \phi_i \rightarrow \Fclo{\FBC_{i \in \FNI'} \FAA_i} \FBW_{i \in \FNI'} \phi_i$.
Clearly, $\vdash_\FXL \FBW_{i \in \FNI} \Fclo{\FAA_i} \phi_i \rightarrow \FBW_{i \in \FNI'} \Fclo{\FAA_i} \phi_i$, and $\vdash_\FXL \Fclo{\FBB_j} \psi_{j} \rightarrow \FBV_{j \in \FPI} \Fclo{\FBB_j} \psi_j$.
Then $\vdash_\FXL \FBW_{i \in \FNI} \Fclo{\FAA_i} \phi_i \rightarrow \FBV_{j \in \FPI} \Fclo{\FBB_j} \psi_j$.
Then, $\vdash_\FXL \gamma \lor (\FBW_{i \in \FNI} \Fclo{\FAA_i} \phi_i \rightarrow \FBV_{j \in \FPI} \Fclo{\FBB_j} \psi_j)$.

\item

Assume $\FSX \in \FXX$ and $\FIX \notin \FXX$.
Note in this case, $\mathtt{A}\text{-}\mathtt{Ser}$ is an axiom of $\FXL$, and $\FNI'$ is $\FSX$-neat: for all $i,i'\in \FNI'$, if $\FAA_{i}\neq \emptyset $ and $\FAA_{i'} \neq \emptyset$, then $i=i'$. 

Assume $\FNI'$ is empty. Note in this case, $\Fclo{\FBC_{i \in \FNI'} \FAA_i} \FBW_{i \in \FNI'} \phi_i \rightarrow \Fclo{\FBB_j} \psi_{j}$ is $\Fclo{\emptyset} \top \rightarrow \Fclo{\FBB_j} \psi_{j}$.
Then $\vdash_\FXL \Fclo{\emptyset} \top \rightarrow \Fclo{\FBB_j} \psi_{j}$.
By Axiom $\mathtt{A}\text{-}\mathtt{Ser}$, we have $\vdash_\FXL \Fclo{\FBB_j} \psi_{j}$. Then $\vdash_\FXL \FBV_{j \in \FPI} \Fclo{\FBB_j} \psi_j$. Then $\vdash_\FXL \FBW_{i \in \FNI} \Fclo{\FAA_i} \phi_i \rightarrow \FBV_{j \in \FPI} \Fclo{\FBB_j} \psi_j$. Then, $\vdash_\FXL \gamma \lor (\FBW_{i \in \FNI} \Fclo{\FAA_i} \phi_i \rightarrow \FBV_{j \in \FPI} \Fclo{\FBB_j} \psi_j)$.

Assume $\FNI'$ is not empty.
By Axiom $\mathtt{A}\text{-}\mathtt{SIA}$, $\vdash_\FXL \FBW_{i \in \FNI'} \Fclo{\FAA_i} \phi_i \rightarrow \Fclo{\FBC_{i \in \FNI'} \FAA_i} \FBW_{i \in \FNI'} \phi_i$.
Clearly, $\vdash_\FXL \FBW_{i \in \FNI} \Fclo{\FAA_i} \phi_i \rightarrow \FBW_{i \in \FNI'} \Fclo{\FAA_i} \phi_i$, and $\vdash_\FXL \Fclo{\FBB_j} \psi_{j} \rightarrow \FBV_{j \in \FPI} \Fclo{\FBB_j} \psi_j$.
Then, $\vdash_\FXL \FBW_{i \in \FNI} \Fclo{\FAA_i} \phi_i$ $\rightarrow \FBV_{j \in \FPI} \Fclo{\FBB_j} \psi_j$.
Then, $\vdash_\FXL \gamma \lor (\FBW_{i \in \FNI} \Fclo{\FAA_i} \phi_i \rightarrow \FBV_{j \in \FPI} \Fclo{\FBB_j} \psi_j)$.

\item

Assume $\FSX \in \FXX$ and $\FIX \in \FXX$.
Note in this case, Axioms $\mathtt{A}\text{-}\mathtt{Ser}$ and $\mathtt{A}\text{-}\mathtt{IA}$ are axioms of $\FXL$, and $\FNI'$ is $\FSX\FIX$-neat: for all $i,i'\in \FNI'$, if $i\neq i'$, then $\FAA_{i}\cap \FAA_{i'}= \emptyset$.

Assume $\FNI'$ is empty. Note in this case, $\Fclo{\FBC_{i \in \FNI'} \FAA_i} \FBW_{i \in \FNI'} \phi_i \rightarrow \Fclo{\FBB_j} \psi_{j}$ is $\Fclo{\emptyset} \top \rightarrow \Fclo{\FBB_j} \psi_{j}$.
Then $\vdash_\FXL \Fclo{\emptyset} \top \rightarrow \Fclo{\FBB_j} \psi_{j}$.
By Axiom $\mathtt{A}\text{-}\mathtt{Ser}$, we have $\vdash_\FXL \Fclo{\FBB_j} \psi_{j}$. Then $\vdash_\FXL \FBV_{j \in \FPI} \Fclo{\FBB_j} \psi_j$.
Then $\vdash_\FXL \FBW_{i \in \FNI} \Fclo{\FAA_i} \phi_i \rightarrow \FBV_{j \in \FPI} \Fclo{\FBB_j} \psi_j$. Then, $\vdash_\FXL \gamma \lor (\FBW_{i \in \FNI} \Fclo{\FAA_i} \phi_i \rightarrow \FBV_{j \in \FPI} \Fclo{\FBB_j} \psi_j)$.

Assume $\FNI'$ is not empty.
By Axiom $\mathtt{A}\text{-}\mathtt{IA}$, $\vdash_\FXL \FBW_{i \in \FNI'} \Fclo{\FAA_i} \phi_i \rightarrow \Fclo{\FBC_{i \in \FNI'} \FAA_i} \FBW_{i \in \FNI'} \phi_i$.
Clearly, $\vdash_\FXL \FBW_{i \in \FNI} \Fclo{\FAA_i} \phi_i \rightarrow \FBW_{i \in \FNI'} \Fclo{\FAA_i} \phi_i$, and $\vdash_\FXL \Fclo{\FBB_j} \psi_{j} \rightarrow \FBV_{j \in \FPI} \Fclo{\FBB_j} \psi_j$.
Then, $\vdash_\FXL \FBW_{i \in \FNI} \Fclo{\FAA_i} \phi_i$ $\rightarrow \FBV_{j \in \FPI} \Fclo{\FBB_j} \psi_j$.
Then, $\vdash_\FXL \gamma \lor (\FBW_{i \in \FNI} \Fclo{\FAA_i} \phi_i \rightarrow \FBV_{j \in \FPI} \Fclo{\FBB_j} \psi_j)$.

\end{enumerate}

\item 

Assume $\FDX \in \FXX$. Then then there is $\FNI'\subseteq \FNI$ and $j \in \FPI$ such that $\FNI'$ is $\FXX$-neat, $\FBC_{i\in \FNI'}\FAA_i\subseteq \FBB_{j}$, and $\vdash_\FXL \FBW_{i\in \FNI'}\phi_i \to (\psi_{j}\vee \FBV_{k \in \FPI_0} \psi_k)$.
By Rule $\mathtt{Mon}$, $\vdash_\FXL \Fclo{\FBC_{i \in \FNI'} \FAA_i} \FBW_{i \in \FNI'} \phi_i \rightarrow \Fclo{\FBB_j} (\psi_{j} \lor \FBV_{k \in \FPI_0} \psi_k)$. Note in this case, Axiom $\mathtt{A}\text{-}\mathtt{Det}$ is an axiom of $\FXL$.
We respectively consider whether $\FSX$ is in $\FXX$ and whether $\FIX$ is in $\FXX$. 

\begin{enumerate}

\item 

Assume $\FSX \notin \FXX$ and $\FIX \notin \FXX$.
Note in this case, $\FNI'$ is $\FDX$-neat: $\FNI'$ is not empty, and for all $i,i'\in \FNI'$, if $\FAA_{i}\neq \emptyset $ and $\FAA_{i'}\neq \emptyset$, then $i=i'$.
By repeated applications of Axiom $\mathtt{A}\text{-}\mathtt{SIA}$, we can get $\vdash_\FXL \FBW_{i \in \FNI'} \Fclo{\FAA_i} \phi_i \rightarrow \Fclo{\FBC_{i \in \FNI'} \FAA_i} \FBW_{i \in \FNI'} \phi_i$.
By repeated applications of Axiom $\mathtt{A}\text{-}\mathtt{Det}$, we can get $\vdash_\FXL \Fclo{\FBB_j} (\psi_{j} \lor \FBV_{k \in \FPI_0} \psi_k) \rightarrow (\Fclo{\FBB_j} \psi_{j} \lor \FBV_{k \in \FPI_0} \Fclo{\FAG} \psi_k)$.
Clearly, $\vdash_\FXL \FBW_{i \in \FNI} \Fclo{\FAA_i} \phi_i \rightarrow \FBW_{i \in \FNI'} \Fclo{\FAA_i} \phi_i$, and $\vdash_\FXL (\Fclo{\FBB_j} \psi_{j} \lor \FBV_{k \in \FPI_0} \Fclo{\FAG} \psi_k) \rightarrow \FBV_{j \in \FPI} \Fclo{\FBB_j} \psi_j$.
Then $\vdash_\FXL \FBW_{i \in \FNI} \Fclo{\FAA_i} \phi_i \rightarrow \FBV_{j \in \FPI} \Fclo{\FBB_j} \psi_j$.
Then, $\vdash_\FXL \gamma \lor (\FBW_{i \in \FNI} \Fclo{\FAA_i} \phi_i \rightarrow \FBV_{j \in \FPI} \Fclo{\FBB_j} \psi_j)$.

\item 

Assume $\FSX \notin \FXX$ and $\FIX \in \FXX$.
Note in this case, Axiom $\mathtt{A}\text{-}\mathtt{IA}$ is an axiom of $\FXL$, and $\FNI'$ is $\FIX\FDX$-neat: $\FNI'$ is not empty, and for all $i,i'\in \FNI'$, if $i\neq i'$, then $\FAA_{i}\cap \FAA_{i'}= \emptyset$.
By repeated applications of Axiom $\mathtt{A}\text{-}\mathtt{IA}$, we can get $\vdash_\FXL \FBW_{i \in \FNI'} \Fclo{\FAA_i} \phi_i \rightarrow \Fclo{\FBC_{i \in \FNI'} \FAA_i} \FBW_{i \in \FNI'} \phi_i$.
By repeated applications of Axiom $\mathtt{A}\text{-}\mathtt{Det}$, we can get $\vdash_\FXL \Fclo{\FBB_j} (\psi_{j} \lor \FBV_{k \in \FPI_0} \psi_k) \rightarrow (\Fclo{\FBB_j} \psi_{j} \lor \FBV_{k \in \FPI_0} \Fclo{\FAG} \psi_k)$.
Clearly, $\vdash_\FXL \FBW_{i \in \FNI} \Fclo{\FAA_i} \phi_i \rightarrow \FBW_{i \in \FNI'} \Fclo{\FAA_i} \phi_i$, and $\vdash_\FXL (\Fclo{\FBB_j} \psi_{j} \lor \FBV_{k \in \FPI_0} \Fclo{\FAG} \psi_k) \rightarrow \FBV_{j \in \FPI} \Fclo{\FBB_j} \psi_j$.
Then $\vdash_\FXL \FBW_{i \in \FNI} \Fclo{\FAA_i} \phi_i \rightarrow \FBV_{j \in \FPI} \Fclo{\FBB_j} \psi_j$.
Then, $\vdash_\FXL \gamma \lor (\FBW_{i \in \FNI} \Fclo{\FAA_i} \phi_i \rightarrow \FBV_{j \in \FPI} \Fclo{\FBB_j} \psi_j)$.

\item

Assume $\FSX \in \FXX$ and $\FIX \notin \FXX$.
Note in this case, Axiom $\mathtt{A}\text{-}\mathtt{Ser}$ is an axiom of $\FXL$, and $\FNI'$ is $\FSX\FDX$-neat: for all $i,i'\in \FNI'$, if $\FAA_{i}\neq \emptyset $ and $\FAA_{i'} \neq \emptyset$, then $i=i'$.

Assume $\FNI'$ is empty. Note in this case, $\Fclo{\FBC_{i \in \FNI'} \FAA_i} \FBW_{i \in \FNI'} \phi_i \rightarrow \Fclo{\FBB_j} (\psi_{j} \lor \FBV_{k \in \FPI_0} \psi_k)$ is $\Fclo{\emptyset} \top \rightarrow \Fclo{\FBB_j} (\psi_{j} \lor \FBV_{k \in \FPI_0} \psi_k)$.
Then $\vdash_\FXL \Fclo{\emptyset} \top \rightarrow \Fclo{\FBB_j} (\psi_{j} \lor \FBV_{k \in \FPI_0} \psi_k)$.
By Axiom $\mathtt{A}\text{-}\mathtt{Ser}$, $\vdash_\FXL \Fclo{\FBB_j} (\psi_{j} \lor \FBV_{k \in \FPI_0} \psi_k)$.
Then $\vdash_\FXL \FBV_{j \in \FPI} \Fclo{\FBB_j} \psi_j$.
Then $\vdash_\FXL \FBW_{i \in \FNI} \Fclo{\FAA_i} \phi_i \rightarrow \FBV_{j \in \FPI} \Fclo{\FBB_j} \psi_j$.
Then, $\vdash_\FXL \gamma \lor (\FBW_{i \in \FNI} \Fclo{\FAA_i} \phi_i \rightarrow \FBV_{j \in \FPI} \Fclo{\FBB_j} \psi_j)$.

Assume $\FNI'$ is not empty.
By repeated applications of Axiom $\mathtt{A}\text{-}\mathtt{SIA}$, we can get $\vdash_\FXL \FBW_{i \in \FNI'} \Fclo{\FAA_i} \phi_i \rightarrow \Fclo{\FBC_{i \in \FNI'} \FAA_i} \FBW_{i \in \FNI'} \phi_i$.
By repeated applications of Axiom $\mathtt{A}\text{-}\mathtt{Det}$, we can get $\vdash_\FXL \Fclo{\FBB_j} (\psi_{j} \lor \FBV_{k \in \FPI_0} \psi_k) \rightarrow (\Fclo{\FBB_j} \psi_{j} \lor \FBV_{k \in \FPI_0} \Fclo{\FAG} \psi_k)$.
Clearly, $\vdash_\FXL \FBW_{i \in \FNI} \Fclo{\FAA_i} \phi_i \rightarrow \FBW_{i \in \FNI'} \Fclo{\FAA_i} \phi_i$, and $\vdash_\FXL (\Fclo{\FBB_j} \psi_{j} \lor \FBV_{k \in \FPI_0} \Fclo{\FAG} \psi_k) \rightarrow \FBV_{j \in \FPI} \Fclo{\FBB_j} \psi_j$.
Then $\vdash_\FXL \FBW_{i \in \FNI} \Fclo{\FAA_i} \phi_i \rightarrow \FBV_{j \in \FPI} \Fclo{\FBB_j} \psi_j$.
Then, $\vdash_\FXL \gamma \lor (\FBW_{i \in \FNI} \Fclo{\FAA_i} \phi_i \rightarrow \FBV_{j \in \FPI}$ $\Fclo{\FBB_j} \psi_j)$.

\item

Assume $\FSX \in \FXX$ and $\FIX \in \FXX$.
Note in this case, Axioms $\mathtt{A}\text{-}\mathtt{Ser}$ and $\mathtt{A}\text{-}\mathtt{IA}$ are axioms of $\FXL$, and $\FNI'$ is $\FSX\FIX\FDX$-neat: for all $i,i'\in \FNI'$, if $i\neq i'$, then $\FAA_{i}\cap \FAA_{i'}= \emptyset$.

Assume $\FNI'$ is empty. Note in this case, $\Fclo{\FBC_{i \in \FNI'} \FAA_i} \FBW_{i \in \FNI'} \phi_i \rightarrow \Fclo{\FBB_j} (\psi_{j} \lor \FBV_{k \in \FPI_0} \psi_k)$ is $\Fclo{\emptyset} \top \rightarrow \Fclo{\FBB_j} (\psi_{j} \lor \FBV_{k \in \FPI_0} \psi_k)$.
Then $\vdash_\FXL \Fclo{\emptyset} \top \rightarrow \Fclo{\FBB_j} (\psi_{j} \lor \FBV_{k \in \FPI_0} \psi_k)$. 
By Axiom $\mathtt{A}\text{-}\mathtt{Ser}$, $\vdash_\FXL \Fclo{\FBB_j} (\psi_{j} \lor \FBV_{k \in \FPI_0} \psi_k)$.
Then $\vdash_\FXL \FBV_{j \in \FPI} \Fclo{\FBB_j} \psi_j$.
Then, $\vdash_\FXL \FBW_{i \in \FNI} \Fclo{\FAA_i} \phi_i \rightarrow \FBV_{j \in \FPI} \Fclo{\FBB_j} \psi_j$.
Then, $\vdash_\FXL \gamma \lor (\FBW_{i \in \FNI} \Fclo{\FAA_i} \phi_i \rightarrow \FBV_{j \in \FPI} \Fclo{\FBB_j} \psi_j)$.

Assume $\FNI'$ is not empty.
By repeated applications of Axiom $\mathtt{A}\text{-}\mathtt{IA}$, we can get $\vdash_\FXL \FBW_{i \in \FNI'} \Fclo{\FAA_i} \phi_i \rightarrow \Fclo{\FBC_{i \in \FNI'} \FAA_i} \FBW_{i \in \FNI'} \phi_i$.
By repeated applications of Axiom $\mathtt{A}\text{-}\mathtt{Det}$, we can get $\vdash_\FXL \Fclo{\FBB_j} (\psi_{j} \lor \FBV_{k \in \FPI_0} \psi_k) \rightarrow (\Fclo{\FBB_j} \psi_{j} \lor \FBV_{k \in \FPI_0} \Fclo{\FAG} \psi_k)$.
Clearly, $\vdash_\FXL \FBW_{i \in \FNI} \Fclo{\FAA_i} \phi_i \rightarrow \FBW_{i \in \FNI'} \Fclo{\FAA_i} \phi_i$, and $\vdash_\FXL (\Fclo{\FBB_j} \psi_{j} \lor \FBV_{k \in \FPI_0} \Fclo{\FAG} \psi_k) \rightarrow \FBV_{j \in \FPI} \Fclo{\FBB_j} \psi_j$.
Then $\vdash_\FXL \FBW_{i \in \FNI} \Fclo{\FAA_i} \phi_i \rightarrow \FBV_{j \in \FPI} \Fclo{\FBB_j} \psi_j$.
Then, $\vdash_\FXL \gamma \lor (\FBW_{i \in \FNI} \Fclo{\FAA_i} \phi_i \rightarrow \FBV_{j \in \FPI}$ $\Fclo{\FBB_j} \psi_j)$.

\end{enumerate}

\end{enumerate}

\end{proof}

\subsection{Completeness by induction}

\begin{theorem}[Soundness and completeness of $\FXL$]
For any $\FXX \in \FES$, the axiomatic system for $\FXL$ given in Definition \ref{definition:Axiomatic systems for XL} is sound and complete with respect to the set of $\FXX$-valid formulas in $\Phi$.
\end{theorem}

\begin{proof}
~

Let $\FXX \in \FES$. The soundness of the system for $\FXL$ is easy to show, and we skip its proof.
The completeness is shown as follows.

Let $\phi \in \Phi$. Assume $\models_\FXL \phi$. We want to show $\vdash_\FXL \phi$. We put an induction on the modal depth of $\phi$.

\textbf{Assume the modal depth of $\phi$ is $0$.}

Then $\phi$ is a formula of the classical propositional logic. As $\FXL$ extends the classical propositional logic, $\vdash_\FXL \phi$.

\textbf{Assume the modal depth of $\phi$ is $n$, which is greater than $0$.}

By Lemma \ref{lemma:normal-form}, there is $\phi'$ such that (1) $\vdash_\FXL \phi \leftrightarrow \phi'$, (2) $\phi'$ has the modal depth $n$, and (3) $\phi'$ is in the form of $\chi_0 \land \dots \land \chi_k$, where every $\chi_i$ is a standard disjunction.
By soundness, $\vDash_\FXL \phi \leftrightarrow \phi'$. Then $\vDash_\FXL \phi'$.
Let $i \leq k$. Then $\models_\FXL \chi_i$. It suffices to show $\vdash_\FXL \chi_i$.

Assume the modal degree of $\chi_i$ is less than $n$. By the inductive hypothesis, $\vdash_\FXL \chi_i$.

Assume the modal degree of $\chi_i$ is $n$. Let $\chi_i = \gamma \lor (\FBW_{i \in \FNI} \Fclo{\FAA_i} \phi_i \rightarrow \FBV_{j \in \FPI} \Fclo{\FBB_j} \psi_j)$.

By the downward validity lemma, the $\FXX$-validity-reduction condition of $\chi_i$ is met:
\begin{enumerate}[label=(\arabic*),leftmargin=3.33em]

\item 

$\models_{\FXL} \gamma$, or 

\item 

the following two conditions hold:
\begin{enumerate}

\item 

if $\mathsf{D} \notin \FXX$, then there is $\FNI' \subseteq \FNI$ and $j \in \FPI$ such that $\FNI'$ is $\FXX$-neat, $\FBC_{i\in \FNI'} \FAA_i\subseteq \FBB_j$, and $\models_\FXL \FBW_{i\in \FNI'}\phi_i \to \psi_j$;

\item 

if $\mathsf{D} \in \FXX$, then there is $\FNI'\subseteq \FNI$ and $j\in \FPI$ such that $\FNI'$ is $\FXX$-neat, $\FBC_{i\in \FNI'} \FAA_i \subseteq \FBB_{j}$, and $\models_\FXL \FBW_{i \in \FNI'} \phi_i \to (\psi_{j} \lor \FBV_{k \in \FPI_0} \psi_k)$.

\end{enumerate}

\end{enumerate}

By the inductive hypothesis, the $\FXX$-derivability-reduction condition of $\chi_i$ is met:
\begin{enumerate}[label=(\arabic*),leftmargin=3.33em]

\item 

$\vdash_{\FXL} \gamma$, or 

\item 

the following two conditions hold:
\begin{enumerate}

\item 

if $\mathsf{D} \notin \FXX$, then there is $\FNI' \subseteq \FNI$ and $j \in \FPI$ such that $\FNI'$ is $\FXX$-neat, $\FBC_{i\in \FNI'} \FAA_i\subseteq \FBB_j$, and $\vdash_\FXL \FBW_{i\in \FNI'}\phi_i \to \psi_j$;

\item 

if $\mathsf{D} \in \FXX$, then there is $\FNI'\subseteq \FNI$ and $j\in \FPI$ such that $\FNI'$ is $\FXX$-neat, $\FBC_{i\in \FNI'} \FAA_i \subseteq \FBB_{j}$, and $\vdash_\FXL \FBW_{i \in \FNI'} \phi_i \to (\psi_{j} \lor \FBV_{k \in \FPI_0} \psi_k)$.

\end{enumerate}

\end{enumerate}

By the upward derivability lemma, $\vdash_\FXL \chi_i$.

\end{proof}

\section{Further work}
\label{section:Further work}

There are two important topics for future work.

$\FCL$ is decidable and the satisfiability problem for it is $\mathrm{PSPACE}$~\cite{pauly_modal_2002}. By employing the downward validity lemma and the upward derivability lemma, it can be shown that the other seven coalition logics are also decidable. What are their computational complexities? We leave this as further work.

$\FATL$ is a temporal extension of $\FCL$, whose featured formulas are as follows:
\begin{itemize}

\item 

$\Fclo{\FAA} \mathsf{X} \phi$: some joint strategy of $\FAA$ ensures $\phi$ at the next moment.

\item 

$\Fclo{\FAA} \mathsf{G} \phi$: some joint strategy of $\FAA$ ensures $\phi$ at all moments in the future.

\item 

$\Fclo{\FAA} (\phi \mathsf{U} \psi)$: some joint strategy of $\FAA$ ensures $\phi$ until $\psi$. 

\end{itemize}

\noindent $\FATL$ is a very useful logic. What are the logics that are temporal extensions of the other seven coalition logics? We will study this in the future.

\subsection*{Acknowledgments}

Thanks go to the audience of seminars at the University of Toulouse and the University of Bergen, and the Third International Workshop on Logic of Multi-agent Systems at Zhejiang University, especially Thomas \r{A}gotnes, Rustam Galimullin, and Emiliano Lorini.

\bibliographystyle{alpha}
\bibliography{Strategy-reasoning,Strategy-reasoning-special}

\end{document}